\documentclass[12pt]{article}

\usepackage{a4wide, amsmath,amsthm,amsfonts,amscd,amssymb,eucal,bbm,mathrsfs}

\def\RR{{\mathbb R}}
\def\CC{{\mathbb C}}

\def\ZZ{{\mathbb Z}}

\def\SS{{\mathbb S}}

\def\A{{\mathcal A}}

\def\H{{\mathcal H}}
\def\I{{\mathcal I}}

\def\M{{\mathcal M}}
\def\N{{\mathcal N}}

\def\P{{\mathcal P}}
\def\R{{\mathcal R}}

\def\T{{\mathcal T}}

\def\W{{\mathcal W}}

\def\a{\alpha}

\def\d{\delta}
\def\e{\varepsilon}
\def\f{\varphi}

\def\G{\Gamma}

\def\k{\kappa}

\def\L{{\mathrm L}}

\def\R{{\mathrm R}}
\def\s{\sigma}
\def\sc{\Sigma}
\def\t{\tau}

\def\Ad{{\hbox{\rm Ad\,}}}

\def\sp{{\rm sp}\,}

\def\1{{\mathbbm 1}}

\def\u1{{\A^{(0)}}}

\def\diff{{\rm Diff}}

\def\diffs1{\diff(S^1)}
\def\vir{{\rm Vir}}

\def\supp{{\rm supp}}

\def\slim{{{\mathrm{s}\textrm{-}\lim}}}

\def\tin{{\mathrm{in}}}
\def\tout{{\mathrm{out}}}

\def\timesi{{\overset{\tin}\times}}
\def\timeso{{\overset{\tout}\times}}
\def\psl2r{{\rm PSL}(2,\RR)}
\def\2dmob{{\overline{\psl2r}\times\overline{\psl2r}}}
\def\<{\langle}
\def\>{\rangle}

\def\poincare{{\P^\uparrow_+}}

\newtheorem{theorem}{Theorem}[section]

\newtheorem{corollary}[theorem]{Corollary}
\newtheorem{proposition}[theorem]{Proposition}
\newtheorem{lemma}[theorem]{Lemma}
\theoremstyle{remark}
\newtheorem{remark}[theorem]{Remark}
\newtheorem{example}[theorem]{Example}

\title{Construction of wedge-local nets of observables
through Longo-Witten endomorphisms}
\date{}
\author{
{\bf Yoh Tanimoto\footnote{Supported in part by the ERC Advanced Grant 227458
OACFT ``Operator Algebras and Conformal Field Theory''.}}\\
Dipartimento di Matematica, Universit\`a di Roma ``Tor
Vergata''\\ Via della Ricerca Scientifica, 1 - I--00133 Roma, Italy.\\
e-mail: {\tt tanimoto@mat.uniroma2.it}}

\begin{document}
\maketitle
\begin{abstract}
A convenient framework to treat massless two-dimensional scattering theories
has been established by Buchholz. In this framework, we show that
the asymptotic algebra and the scattering matrix completely characterize
the given theory under asymptotic completeness and standard assumptions.

Then we obtain several families of interacting wedge-local nets
by a purely von Neumann algebraic procedure.
One particular case of them coincides with the deformation of chiral
CFT by Buchholz-Lechner-Summers. In another case, we manage to determine
completely the strictly local elements.
Finally, using Longo-Witten endomorphisms on the $U(1)$-current net
and the free fermion net, a large family of wedge-local nets is constructed.
\end{abstract}

\section{Introduction}\label{introduction}
Construction of interacting models of quantum field theory in physical
four-dimensional spacetime has been a long-standing open problem since 
the birth of quantum theory. Recently, operator-algebraic methods have
been applied to construct models with weaker localization property \cite{GL07, GL08, BS08, BLS, Lechner11}.
It is still possible to calculate the two-particle scattering matrix for these weakly
localized theories and they have been shown to be nontrivial.
However, the strict locality still remains difficult. Indeed, of these deformed theories,
strictly localized contents have been shown to be trivial in higher dimensions \cite{BLS}.
In contrast, in two-dimensional spacetime, a family of strictly local theories
has been constructed and nontrivial scattering matrices have been calculated \cite{Lechner08}.
The construction of local nets of observables is split up into two procedures:
construction of wedge-local nets and determination of strictly local elements.
In this paper we present a purely von Neumann algebraic procedure to construct
wedge-local nets based on chiral CFT and completely determine strictly local elements
for some of these wedge-local nets. Furthermore, we show that the pair of the S-matrix
and the asymptotic algebra forms a complete invariant of the given net and
give a simple formula to recover the original net from these data.

In algebraic approach to quantum field theory, or algebraic QFT, theories are realized
as local nets of operator algebras. Principal examples are constructed from local quantum fields,
or in mathematical terms, from operator-valued distributions which commute in spacelike regions.
However, recent years purely operator-algebraic constructions of such nets have been found.
A remarkable feature of these new constructions is that they first consider a single von Neumann algebra
(instead of a family of von Neumann algebras) which
is acted on by the spacetime symmetry group in an appropriate way.
The construction procedure relying on a single von Neumann algebra has been proposed
in \cite{Borchers92} and resulted in some intermediate constructions \cite{GL07, GL08, BLS, Lechner11}
and even in a complete construction of local nets \cite{Lechner08}.
This von Neumann algebra is
interpreted as the algebra of observables localized in a wedge-shaped region. There is
a prescription to recover the strictly localized observables \cite{Borchers92}.
However, the algebras of strictly localized observables are not necessarily large
enough and they can be even trivial \cite{BLS}. When it turned out to be sufficiently
large, one had to rely on the modular nuclearity condition, a sophisticated analytic
tool \cite{BDL90, Lechner08}.

Among above constructions, the deformation by Buchholz, Lechner and Summers
starts with an arbitrary wedge-local net.
When one applies the BLS deformation to chiral conformal theories in two dimensions,
things get considerably simplified. We have seen that the theory remains to be
asymptotically complete in the sense of waves \cite{Buchholz75} even after the deformation and the full
S-matrix has been computed \cite{DT11-1}. In this paper we carry out a further
construction of wedge-local nets based on chiral conformal nets.
It turns out that all these construction are related with endomorphisms of
the half-line algebra in the chiral components recently studied by Longo and Witten \cite{LW11}.
Among such endomorphisms, the simplest ones are translations and inner symmetries.
We show that the construction related to translations coincides with the BLS deformation
of chiral CFT. The construction related to inner symmetries is new and we completely determine
the strictly localized observables under some technical conditions.
Furthermore, by using the family of endomorphisms on the $U(1)$-current net considered in \cite{LW11},
we construct a large family of wedge-local nets parametrized by inner symmetric functions.
All these wedge-local nets have nontrivial S-matrix, but the strictly local part of
the wedge-local nets constructed through inner symmetries has trivial S-matrix.
The strict locality of the other constructions remains open.
Hence, to our opinion, the true difficulty lies in strict locality.

Another important question is how large the class of theories is obtained by
this procedure. The class of S-matrices so far obtained is considered
rather small, since any of such S-matrices is contained in the tensor product of
abelian algebras in chiral components, which corresponds to the notion of
local diagonalizability in quantum information.
In this paper, however, we show that a massless asymptotically complete
theory is completely characterized by its asymptotic behaviour and the
S-matrix, and the whole theory can be recovered with a simple formula.
Hence we can say that this formula is sufficiently general.

In Section \ref{preliminaries} we recall standard notions of algebraic QFT and
scattering theory. In Section \ref{asymptotic-chiral-algebra} we show that the pair of S-matrix
and the asymptotic algebra is a complete invariant of a massless asymptotically complete net.
In Section \ref{one-parameter} we construct wedge-local nets using one-parameter endomorphisms of Longo-Witten.
It is shown that the case of translations coincides with the BLS deformation of chiral CFT and
the strictly local elements are completely determined for the case of inner symmetries.
A common argument is summarized in Section \ref{commutativity}.
Section \ref{u1-current} is devoted to the construction of wedge-local nets based
on a specific example, the $U(1)$-current net. A similar construction is obtained
also for the free fermionic net.
Section \ref{conclusion} summarizes our perspectives.

\section{Preliminaries}\label{preliminaries}
\subsection{Poincar\'e covariant net}\label{poincare}
We recall the algebraic treatment of quantum field theory \cite{Haag}.
A {\bf (local) Poincar\'e covariant net $\A$ on $\RR^2$} assigns to each open bounded
region $O$ a von Neumann algebra $\A(O)$ on a (separable) Hilbert space $\H$
satisfying the following conditions:
\begin{enumerate}
\item[(1)] {\bf Isotony.} If $O_1 \subset O_2$, then $\A(O_1) \subset \A(O_2)$.
\item[(2)] {\bf Locality.} If $O_1$ and $O_2$ are causally disjoint, then $[\A(O_1),\A(O_2)] = 0$.
\item[(3)] {\bf Poincar\'e covariance.} There exists a strongly continuous unitary
representation $U$ of the (proper orthochronous) Poincar\'e group $\poincare$ such that
for any open region $O$ it holds that
\begin{equation*}
U(g)\A(O)U(g)^* = \A(gO), \mbox{ for } g \in \poincare.
\end{equation*}
\item[(4)]{\bf Positivity of energy.} The joint spectrum of
the translation subgroup $\RR^2$ of $\poincare$
of $U$ is contained in the forward lightcone $V_+ = \{(p_0,p_1) \in \RR^2: p_0 \ge |p_1|\}$.
\item[(5)] {\bf Existence of the vacuum.} There is a unique (up to a phase) unit vector $\Omega$ in
$\H$ which is invariant under the action of $U$,
and cyclic for $\bigvee_{O \Subset \RR^2} \A(O)$.
\item[(6)] {\bf Additivity.} If $O = \bigcup_i O_i$, then $\A(O) = \bigvee_i \A(O_i)$.
\end{enumerate}

 From these axioms, the following property automatically follows (see \cite{Baumgaertel})
\begin{enumerate}
\item[(7)] {\bf Reeh-Schlieder property.} The vector $\Omega$ is cyclic and separating for each $\A(O)$.
\end{enumerate}

It is convenient to extend the definition of net also to a class of unbounded
regions called {\bf wedges}.
By definition, the standard left and right wedges are as follows:
\begin{eqnarray*}
W_\L &:=& \{(t_0,t_1): t_0 > t_1, t_0 < -t_1\}\\
W_\R &:=& \{(t_0,t_1): t_0 < t_1, t_0 > -t_1\}
\end{eqnarray*}
The wedges $W_\L$, $W_\R$ are invariant under Lorentz boosts.
They are causal complements of each other. All the regions
obtained by translations of standard wedges are still
called left- and right-wedges, respectively.
Moreover, a bounded region obtained as the intersection of a left wedge and a right wedge
is called a {\bf double cone}.
Let $O'$ denote the causal complement of $O$. It holds that $W_\L' = W_\R$, and
if $D = (W_\L+a)\cap(W_\R+b)$ is a double cone, $a,b\in\RR^2$, then
$D' = (W_\R+a)\cup(W_\L+b)$. It is easy to see that $\Omega$ is still cyclic and
separating for $\A(W_\L)$ and $\A(W_\R)$.

We assume the following properties as natural conditions.
\begin{itemize}
\item {\bf Bisognano-Wichmann property.} The modular group $\Delta^{it}$ of $\A(W_\R)$
with respect to $\Omega$ is equal to $U(\Lambda(-2\pi t))$, where
$\Lambda(t) = \left(\begin{matrix} \cosh t & \sinh t\\ \sinh t & \cosh t \end{matrix}\right)$ denotes
the Lorentz boost.
\item {\bf Haag duality.} If $O$ is a wedge or a double cone, then it holds that
$\A(O)' = \A(O')$.
\end{itemize}
If $\A$ is M\"obius covariant (conformal),
then the Bisognano-Wichmann property is automatic
\cite{BGL}, and Haag duality is equivalent to strong additivity
(\cite{Rehren}, see also Section \ref{chiral}).
These properties are valid even in massive interacting models \cite{Lechner08}.
Duality for wedge regions (namely $\A(W_\L)' = \A(W_\R)$) follows from Bisognano-Wichmann
property \cite{Tanimoto11-2}, and it implies that the dual net indeed satisfies
the Haag duality \cite{Baumgaertel}.

\subsection{Chiral conformal nets}\label{chiral}
In this Section we introduce a fundamental class of examples of Poincar\'e covariant nets.
For this purpose, first we explain nets on the one-dimensional circle $S^1$.
An open nonempty connected nondense subset $I$ of the circle $S^1$ is called an interval.
A {\bf (local) M\"obius covariant net $\A_0$ on $S^1$} assigns to each interval
a von Neumann algebra $\A_0(I)$ on a (separable) Hilbert space $\H_0$
satisfying the following conditions:
\begin{enumerate}
\item[(1)] {\bf Isotony.} If $I_1 \subset I_2$, then $\A_0(I_1) \subset \A_0(I_2)$.
\item[(2)] {\bf Locality.} If $I_1 \cap I_2 = \emptyset$, then $[\A_0(I_1),\A_0(I_2)] = 0$.
\item[(3)] {\bf M\"obius covariance.} There exists a strongly continuous unitary
representation $U_0$ of the M\"obius group $\psl2r$ such that
for any interval $I$ it holds that
\begin{equation*}
U_0(g)\A_0(I)U_0(g)^* = \A_0(gI), \mbox{ for } g \in \psl2r.
\end{equation*}
\item[(4)]{\bf Positivity of energy.} The generator of the one-parameter subgroup of
rotations in the representation $U_0$ is positive.
\item[(5)] {\bf Existence of the vacuum.} There is a unique (up to a phase) unit vector $\Omega_0$ in
$\H_0$ which is invariant under the action of $U_0$,
and cyclic for $\bigvee_{I \Subset S^1} \A_0(I)$.
\end{enumerate}
We identify the circle $S^1$ as the one-point compactification of the real line $\RR$
by the Cayley transform:
\[
t = i\frac{z-1}{z+1} \Longleftrightarrow z = -\frac{t-i}{t+i}, \phantom{...}t \in \RR,
\phantom{..}z \in S^1 \subset \CC.
\]
Under this identification, we refer to translations $\t$ and dilations $\d$
of $\RR$ and these are contained in $\psl2r$.
It is known that the positivity of energy is equivalent to the positivity of
the generator of translations \cite{Longo08}.

 From the axioms above, the following properties automatically follow (see \cite{GF93})
\begin{enumerate}
\item[(6)] {\bf Reeh-Schlieder property.} The vector $\Omega_0$ is cyclic and separating for each $\A_0(I)$.
\item[(7)] {\bf Additivity.} If $I = \bigcup_i I_i$, then $\A_0(I) = \bigvee_i \A_0(I_i)$.
\item[(8)] {\bf Haag duality on $S^1$.} For an interval $I$ it holds that $\A_0(I)' = \A_0(I')$,
where $I'$ is the interior of the complement of $I$ in $S^1$.
\item[(9)] {\bf Bisognano-Wichmann property.} The modular group $\Delta_0^{it}$ of $\A_0(\RR_+)$
with respect to $\Omega_0$ is equal to $U_0(\delta(-2\pi t))$, where
$\delta$ is the one-parameter group of dilations.
\end{enumerate}

\begin{example}
At this level, we have a plenty of examples: The simplest one is the $U(1)$-current
net which will be explained in detail in Section \ref{u1-current-summary}.
Among others, the most important family is
the loop group nets \cite{GF93, Wassermann}. Even a classification result has been
obtained for a class of nets on $S^1$ \cite{KL04-1}.
\end{example}

A net $\A_0$ on $S^1$ is said to be {\bf strongly additive} if it holds that
$\A_0(I) = \A_0(I_1)\vee\A_0(I_2)$, where $I_1$ and $I_2$ are intervals obtained
by removing an interior point of $I$.

Let us denote by $\diffs1$ the group of orientation-preserving
diffeomorphisms of the circle $S^1$. This group naturally includes $\psl2r$.
A M\"obius covariant net $\A_0$ on $S^1$ is said to be
{\bf conformal} or {\bf diffeomorphism covariant}
if the representation $U_0$ of $\psl2r$ associated to $\A_0$
extends to a projective unitary representation of $\diffs1$ such that
for any interval $I$ and $x \in \A_0(I)$ it holds that
\begin{gather*}
U_0(g)\A_0(I)U_0(g)^* = \A_0(gI), \mbox{ for } g \in \diffs1,\\
U_0(g)xU_0(g)^* = x, \mbox{ if } \supp(g) \subset I^\prime,
\end{gather*}
where $\supp(g) \subset I^\prime$ means that $g$ acts identically on $I$.

Let $\A_0$ be a M\"obius covariant net on $S^1$. If a unitary operator
$V_0$ commutes with the translation unitaries $T_0(t) = U_0(\t(t))$ and it holds that
$V_0\A_0(\RR_+)V_0^* \subset \A_0(\RR_+)$, then we say that $V_0$ implements
a {\bf Longo-Witten endomorphism} of $\A_0$. In particular, $V_0$ preserves
$\Omega_0$ up to a scalar since $\Omega_0$ is the unique invariant vector
under $T_0(t)$.
Such endomorphisms have been studied first in \cite{LW11} and they found a large family
of endomorphisms for the $U(1)$-current net, its extensions and the free fermion net.

Let us denote two lightlines by $\L_\pm := \{(t_0,t_1)\in\RR^2: t_0\pm t_1 = 0\}$.
Note that any double cone $D$ can be written as a direct product of intervals
$D = I_+\times I_-$ where $I_+ \subset \L_+$ and $I_- \subset \L_-$.
Let $\A_1,\A_2$ be two M\"obius covariant nets on $S^1$ defined on the Hilbert spaces $\H_1,\H_2$ 
with the vacuum vectors $\Omega_1,\Omega_2$ and the representations $U_1,U_2$ of $\psl2r$.
 From this pair, we can construct a two-dimensional net $\A$ as follows:
For a double cone $D = I_+\times I_-$, we set $\A(D) = \A_1(I_+)\otimes\A_2(I_-)$.
For a general open region $O \subset \RR$, we set $\A(O) := \bigvee_{D\subset O} \A(D)$.
We set $\Omega := \Omega_1\otimes\Omega_2$ and
define the representation $U$ of $\psl2r\times\psl2r$
by $U(g_1\times g_2) := U_1(g_1)\otimes U_2(g_2)$.
By recalling that $\psl2r\times\psl2r$ contains the Poincar\'e group $\poincare$,
it is easy to see that $\A$ together with $U$ and $\Omega$ is a Poincar\'e covariant net.
We say that such $\A$ is {\bf chiral} and $\A_1,\A_2$ are referred to as the {\bf chiral components}.
If $\A_1,\A_2$ are conformal, then the representation $U$ naturally extends to a projective representation
of $\diffs1\times\diffs1$.

\subsection{Scattering theory for Borchers triples}\label{scattering-theory}
A {\bf Borchers triple} on a Hilbert space $\H$ is a triple $(\M, T, \Omega)$ of
a von Neumann algebra $\M \subset B(\H)$, a unitary representation $T$ of $\RR^2$ on $\H$
and a vector $\Omega \in \H$ such that
\begin{itemize}
\item $\Ad T(t_0,t_1)(\M) \subset \M$ for $(t_0,t_1) \in W_\R$, the standard right wedge.
\item The joint spectrum $\sp T$ is contained in the forward lightcone
$V_+ = \{(p_0,p_1) \in \RR_2: p_0 \ge |p_1|\}$.
\item $\Omega$ is a unique (up to scalar) invariant vector under $T$, and cyclic and
separating for $\M$.
\end{itemize}
By the theorem of Borchers \cite{Borchers92,Florig98}, the representation $T$ extends to
the Poincar\'e group $\poincare$, with Lorentz boosts represented by the modular group
of $\M$ with respect to $\Omega$. With this extension $U$, $\M$ is Poincar\'e covariant
in the sense that if $gW_\R \subset W_\R$ for $g\in\poincare$, then $U(g)\M U(g)^* \subset \M$.

The relevance of Borchers triples comes from the fact that
we can construct wedge-local nets from them:
Let $\W$ be the set of wedges, i.e. the the set of all $W=gW_\R$ 
or $W = gW_\L$ where $g$ is a Poincar\'e transformation.
A {\bf wedge-local net} $\W\ni W\mapsto \A(W)$ is a map
from $\W$ to the set of von Neumann algebras which satisfy
isotony, locality, Poincar\'e covariance, positivity of energy,
and existence of vacuum, restricted to $\W$.
A wedge-local net associated 
with the Borchers triple $(\M,T,\Omega)$ is the map defined by 
$\A(W_R+a)=T(a)\M T(a)^\ast$ and $\A(W_R'+a) = T(a)\M' T(a)^\ast$.
This can be considered as a notion of nets with a weaker
localization property. It is clear that there is a one-to-one
correspondence between Borchers triples and wedge-local nets.
A further relation with local nets will be explained at the end
of this section. For simplicity, we study always Borchers triples,
which involve only a single von Neumann algebra.

We denote by $\H_+$ (respectively by $\H_-$) the space of the single excitations with positive momentum,
(respectively with negative momentum) i.e., $\H_+ = \{\xi \in \H: T(t,t)\xi = \xi \mbox{ for } t\in\RR\}$
(respectively $\H_- = \{\xi \in \H: T(t,-t)\xi = \xi \mbox{ for } t\in\RR\}$).

Our fundamental examples come from Poincar\'e covariant nets.
For a Poincar\'e covariant net $\A$, we can construct
a Borchers triple as follows:
\begin{itemize}
\item $\M = \A(W_\R)$
\item $T := U|_{\RR^2}$, the restriction of $U$ to the translation subgroup.
\item $\Omega$: the vacuum vector.
\end{itemize}
Indeed, the first condition follows from the Poincar\'e (in particular, translation) covariance
of the nets and the other conditions are assumed properties of $U$ and $\Omega$ of the net.
If $(\M,T,\Omega)$ comes from a chiral conformal net $\A = \A_1\otimes\A_2$, then
we say this triple is chiral, as well.
This simple construction by tensor product of chiral nets
is considered to be the ``undeformed net''.
We will exhibit later different constructions.

Given a Borchers triple $(\M,T,\Omega)$, we can consider the scattering theory
with respect to massless particles \cite{DT11-1}, which is an extension of \cite{Buchholz75}:
For a bounded operator $x \in B(\H)$ we write $x(a) = \Ad T(a)(x)$ for $a \in \RR^2$.
Furthermore, we define a family of operators parametrized by $\T$:
\[
x_\pm(h_\T) := \int dt\, h_\T(t) x(t,\pm t),
\]
where $h_\T(t) = |\T|^{-\e}h(|\T|^{-\e}(t-\T))$, $0<\e<1$ is a constant,
$\T \in \RR$ and $h$ is a nonnegative symmetric smooth function on $\RR$ such that
$\int dt\, h(t) = 1$.

\begin{lemma}[\cite{Buchholz75} Lemma 2(b), \cite{DT11-1} Lemma 2.1]\label{lm:asymptotic-field}
Let $x \in \M$, then the limits
$\Phi^\tout_+(x) := \underset{\T\to+\infty}\slim \, x_+(h_\T)$ and
$\Phi^\tin_-(x) := \underset{\T\to-\infty}\slim \, x_-(h_\T)$ exist
and it holds that
\begin{itemize}
\item $\Phi^\tout_+(x)\Omega = P_+ x\Omega$ and $\Phi^\tin_-(x)\Omega = P_- x\Omega$
\item $\Phi^\tout_+(x)\H_+ \subset \H_+$ and $\Phi^\tin_-(x)\H_- \subset \H_-$.
\item $\Ad U(g) (\Phi^\tout_+(x)) = \Phi^\tout_+(\Ad U(g)(x))$ and
 $\Ad U(g) (\Phi^\tin_-(x)) = \Phi^\tin_-(\Ad U(g)(x))$ for $g\in\poincare$ such that
 $g W_\R \subset W_\R$.
\end{itemize}
Furthermore, the limits $\Phi^\tout_+(x)$ (respectively $\Phi^\tin_-(x)$)
depends only on $P_+ x\Omega$ (respectively on $P_- x\Omega$).
\end{lemma}

Similarly we define asymptotic objects for the left wedge $W_\L$.
Since $J\M'J=\M$, where $J$ is the modular conjugation for $\M$ with respect to
$\Omega$, we can define for any  $y\in\M'$ 
\begin{equation*}
\Phi_+^\tin(y) := J\Phi_+^\tout(JyJ)J, \,\, \Phi_-^\tout(y):=J\Phi_-^\tin(JyJ)J.
\end{equation*}
Then we have the following.
\begin{lemma}[\cite{DT11-1}, Lemma 2.2]\label{lm:asymptotic-field2}
Let $y\in \M'$.  Then
\begin{equation*}
\Phi_+^\tin(y)=\underset{\T\to-\infty}{\slim}\,y_+(h_\T), \,\,
\Phi_-^\tout(y)=\underset{\T\to\infty}{\slim}\,y_-(h_\T).
\end{equation*}
These operators depend only on the respective vectors $\Phi_+^\tin(y)\Omega = P_+y\Omega$,
$\Phi_-^\tout(y)\Omega = P_- y\Omega$ and we have
\begin{enumerate}
\item[(a)] $\Phi_+^\tin(y)\H_+\subset\H_+, \,\, \Phi_-^\tout(y)\H_-\subset\H_-$,
\item[(b)] $\Ad U(g)(\Phi_+^\tin(y)) = \Phi_+^\tin(\Ad U(g)(y))$ and
$\Ad U(g)(\Phi_-^\tout(y)) = \Phi_-^\tout(\Ad U(g)(y))$ for $g\in\poincare$ such that
$g W_\L \subset W_\L$.
\end{enumerate}
\end{lemma}

For $\xi_+ \in \H_+$, $\xi_- \in \H_-$, there are sequences of local operators $\{x_n\} \subset \M$ and
$\{y_n\} \subset \M'$ such that 
$\underset{n\to\infty}\slim \,P_+ x_n\Omega = \xi_+$ and
$\underset{n\to\infty}\slim \,P_- y_n\Omega = \xi_-$.
With these sequences we define collision states as in \cite{DT11-1}:
\begin{eqnarray*}
\xi_+\timesi\xi_- &=& \underset{n\to\infty}\slim \,\Phi^\tin_+(x_n)\Phi^\tin_-(y_n)\Omega\\
\xi_+\timeso\xi_- &=& \underset{n\to\infty}\slim \,\Phi^\tout_+(x_n)\Phi^\tout_-(y_n)\Omega.
\end{eqnarray*}
We interpret $\xi_+\timesi\xi_-$ (respectively $\xi_+\timeso\xi_-$) as the incoming
(respectively outgoing) state which describes two non-interacting waves
$\xi_+$ and $\xi_-$.
These asymptotic states have the following natural properties.
\begin{lemma}[\cite{DT11-1}, Lemma 2.3]\label{lm:collision-states}
For the collision states $\xi_+\timesi\xi_-$ and $\eta_+\timesi\eta_-$ it holds that
\begin{enumerate}
\item $\<\xi_+\timesi\xi_-, \eta_+\timesi\eta_-\> = \<\xi_+,\eta_+\>\cdot\<\xi_-,\eta_-\>$.
\item $U(g) (\xi_+\timesi\xi_-) = (U(g)\xi_+)\timesi(U(g)\xi_-)$
for all $g\in\poincare$ such that $gW_\R\subset W_\R$.
\end{enumerate}
And analogous formulae hold for outgoing collision states.
\end{lemma}

Furthermore, we set the spaces of collision states: Namely,
we let $\H^\tin$ (respectively $\H^\tout$) be the subspace generated by $\xi_+\timesi\xi_-$
(respectively $\xi_+\timeso\xi_-$).
 From Lemma \ref{lm:collision-states}, we see that the following map
\[
S: \xi_+\timeso\xi_- \longmapsto \xi_+\timesi\xi_-
\]
is an isometry. The operator $S: \H^\tout \to \H^\tin$ is called
the {\bf scattering operator} or the {\bf S-matrix} of the Borchers triple $(\M,U,\Omega)$.
We say the waves in the triple are {\bf interacting} if $S$ is not a constant
multiple of the identity operator on $\H^\tout$.
We say that the Borchers triple is {\bf asymptotically complete (and massless)} if
it holds that $\H^\tin=\H^\tout=\H$.
We have seen that a chiral net and its BLS deformations
(see Section \ref{bls-deformation}) are asymptotically complete \cite{DT11-1}.
If the Borchers triple $(\M,T,\Omega)$ is constructed from a Poincar\'e covariant net $\A$,
then we refer to these objects and notions
as $S,\H_\pm$ and asymptotic completeness of $\A$, etc.
This notion of asymptotic completeness
concerns only massless excitations. Indeed, if one considers
the massive free model for example, then it is easy to see that all the asymptotic
fields are just the vacuum expectation (mapping to $\CC \1$).

To conclude this section, we put a remark on the term ``wedge-local net''.
If a Borchers triple $(\M,T,\Omega)$ comes from
a Haag dual Poincar\'e covariant net $\A$, then
the local algebras are recovered by the formula $\A(D) = T(a)\M T(a)^*\cap T(b)\M' T(b)^* = \A(W_\R+a)\cap\A(W_\L+b)$,
where $D = (W_\R+a)\cap(W_\L+b)$ is a double cone. Furthermore, if $\A$ satisfies
Bisognano-Wichmann property, then the Lorentz boost is obtained from the modular
group, hence all the components of the net are regained from the triple.
Conversely, for a given Borchers triple, one can define a ``local net'' by the same formula
above. In general, this ``net'' satisfies isotony, locality, Poincar\'e covariance
and positivity of energy, but not necessarily satisfies additivity and cyclicity
of vacuum \cite{Borchers92}. Addivity is usually used only in the proof of Reeh-Schlieder property,
thus we do not consider it here. If the ``local net'' constructed from a Borchers triple
satisfies cyclicity of vacuum, we say that the original Borchers triple is {\bf strictly local}.
In this respect, a Borchers triple or a wedge-local net
is considered to have a weaker localization property.
Hence the search for Poincar\'e covariant nets reduces to the search for
strictly local nets. Indeed, by this approach a family of (massive)
interacting Poincar\'e nets has been obtained \cite{Lechner08}.

\section{Asymptotic chiral algebra and S-matrix}\label{asymptotic-chiral-algebra}
\subsection{Complete invariant of nets}\label{complete-invariant}
Here we observe that asymptotically complete (massless) net $\A$ is completely determined by
its behaviour at asymptotic times. This is particularly nice, since the search for Poincar\'e
covariant nets is reduced to the search for appropriate S-matrices. Having seen the classification
of a class of chiral components \cite{KL04-1}, one would hope even for
a similar classification result for massless asymptotically complete nets.


Specifically, we construct a complete invariant of a net with Bisognano-Wichmann
property consisting of two elements.
We already know the first element, the S-matrix. Let us construct the second element,
the asymptotic algebra.
An essential tool is half-sided modular inclusion (see \cite{Wiesbrock93, AZ05} for the original references).
Indeed, we use an analogous argument as in \cite[Lemma 5.5]{Tanimoto11-2}.
Let $\N \subset \M$ be an inclusion of von Neumann algebras. If there is a cyclic and separating
vector $\Omega$ for $\N, \M$ and $\M\cap\N'$, then the inclusion $\N \subset \M$ is said to be
{\bf standard} in the sense of \cite{DL84}.
If $\s^\M_t(\N) \subset \N$ for $t \in \RR_\pm$ where $\s^\M_t$ is the modular
automorphism of $\M$ with respect to $\Omega$, then it is called
a {\bf $\pm$half-sided modular inclusion}.
\begin{theorem}[\cite{Wiesbrock93, AZ05}]
If $(\N\subset \M, \Omega)$ is a standard $+$(respectively $-$)half-sided
modular inclusion, then
there is a M\"obius covariant net $\A_0$ on $S^1$ such that $\A_0(\RR_-) = \M$
and $\A_0(\RR_- - 1) = \N$ (respectively $\A_0(\RR_+) = \M$ and $\A_0(\RR_+ + 1) = \N$).

If a unitary representation $T_0$ of $\RR$ with positive spectrum satisfies
$T_0(t)\Omega = \Omega$ for $t\in\RR$,
$\Ad T_0(t)(\M) \subset \M$ for $t\le 0$ (respectively $t\ge 0$) and $\Ad T_0(-1)(\M) = \N$
(respectively $\Ad T_0(1)(\M) = \N$), then $T_0$ is the representation of the translation
group of the M\"obius covariant net constructed above.
\end{theorem}

We put $\tilde{\A}^\tout_+(O) := \{\Phi^\tout_+(x), x \in \A(O)\}$.
We will show that $\tilde{\A}^\tout_+(W_\R+(-1,1)) \subset \tilde{\A}^\tout_+(W_\R)$
is a standard $+$half-sided modular inclusion when restricted to $\H_+$.
Indeed, $\Phi^\tout_+$ commutes with $\Ad U(g_t)$ where $g_t = \Lambda(-2\pi t)$ is
a Lorentz boost (Lemma \ref{lm:asymptotic-field}),
and $\tilde{\A}(W_\R+(-1,1))$ is sent into itself under $\Ad U(g_t)$ for $t\ge 0$.
Hence by Bisognano-Wichmann property, $\tilde{\A}^\tout_+(W_\R+(-1,1)) \subset \tilde{\A}^\tout_+(W_\R)$
is a $+$half-sided modular inclusion.
In addition, when restricted to $\H_+$, this inclusion is standard.
To see this, note that
$\tilde{\A}^\tout_+(W_\R+(-1,1)) = \tilde{\A}^\tout_+(W_\R+(-1,1)+(1,1)) = \tilde{\A}^\tout_+(W_\R+(0,2))$
because $\Phi^\tout_+$ is invariant under $T(1,1)$,
and hence $\tilde{\A}^\tout_+(D) \subset \left(\tilde{\A}^\tout_+(W_\R+(-1,1))'\cap \tilde{\A}^\tout_+(W_\R)\right)$,
where $D = W_\R\cap (W_\L+(0,2))$.
It follows that
\[
\overline{\left(\tilde{\A}^\tout_+\left(W_\R+(-1,1)\right)' \cap \tilde{\A}^\tout_+(W_\R)\right) \Omega}
\supset \overline{\tilde{\A}^\tout_+(D)\Omega} = \overline{P_+\A(D)\Omega} = \H_+,
\]
which is the standardness on $\H_+$. Then we obtain a M\"obius covariant net on $S^1$ acting on
$\H_+$, which we denote by $\A^\tout_+$.
Similarly we get a M\"obius covariant net $\A^\tout_-$ on $\H_-$.
Two nets $\A^\tout_+$ and $\A^\tout_-$ act like tensor product by Lemma \ref{lm:collision-states},
and span the whole space $\H$ from the vacuum $\Omega$ by asymptotic completeness.
In other words, $\A^\tout_+\otimes\A^\tout_-$ is a chiral M\"obius covariant net on $\RR^2$
acting on $\H$. We call this chiral net $\A^\tout_+\otimes\A^\tout_-$ the
{\bf (out-)asymptotic algebra} of the given net $\A$. Similarly one defines
$\A^\tin_+$ and $\A^\tin_-$.

Let $(\M,T,\Omega)$, where $\M := \A(W_\R)$,
be the Borchers triple associated to an asymptotically complete
Poincare covariant net $\A$ which satisfies Bisognano-Wichmann property and Haag duality.
Our next observation is that $\M$ can be recovered from asymptotic fields.
\begin{proposition}\label{pr:recovery-wedge}
It holds that $\M = \{\Phi^\tout_+(x),\Phi^\tin_-(y): x,y\in \M\}'' =
\tilde{\A}^\tout_+(\RR_-)\vee\tilde{\A}^\tin_-(\RR_+)$.
\end{proposition}
\begin{proof}
The inclusion $\M \supset \{\Phi^\tout_+(x),\Phi^\tin_-(y): x,y\in \M\}''$ is
obvious from the definition of asymptotic fields.
The converse inclusion is established by the modular theory:
 From the assumption of Bisognano-Wichmann property, the modular automorphism
of $\M$ with respect to $\Omega$ is the Lorentz boosts $U(\Lambda(-2\pi t))$.
Furthermore, it holds that
$\Ad U(\Lambda(-2\pi t)) (\Phi^\tout_+(x)) = \Phi^\tout_+(\Ad U(\Lambda(-2\pi t))(x))$
by Lemma \ref{lm:asymptotic-field}. An analogous formula holds for $\Phi^\tin$.
Namely, the algebra in the middle term of the statement
is invariant under the modular group.

By the assumed asymptotic completeness, the algebra in the middle term spans the whole space $\H$
from the vacuum $\Omega$ as well.
Hence by a simple consequence of Takesaki's theorem \cite[Theorem IX.4.2]{TakesakiII}
\cite[Theorem A.1]{Tanimoto11-2}, these two algebras coincide.

The last equation follows by the definition of asymptotic algebra and their
invariance under translations in respective directions.
\end{proof}

\begin{proposition}\label{pr:recovery-field}
It holds that $S\cdot \Phi^\tout_\pm(x) \cdot S^*= \Phi^\tin_\pm(x)$
and $S\cdot \tilde{\A}^\tout_\pm(\RR_\mp) \cdot S^*= \tilde{\A}^\tin_\pm(\RR_\mp)$.
\end{proposition}\begin{proof}
This follows from the calculation, using Lemmata \ref{lm:asymptotic-field}, \ref{lm:asymptotic-field2} and
\ref{lm:collision-states},
\begin{eqnarray*}
\Phi^\tin_+(x) (\xi\timesi \eta)
&=& (P_+x\xi)\timesi\eta \\
&=& S\left((P_+x\xi)\timeso\eta\right) \\
&=& S\cdot \Phi^\tout_+(x)(\xi\timeso\eta) \\
&=& S\cdot \Phi^\tout_+(x)\cdot S^*(\xi\timesi\eta),
\end{eqnarray*}
and asymptotic completeness. The equation for ``$-$'' fields is proved analogously.
The last equalities are simple consequences of the formulae for asymptotic fields.
\end{proof}

\begin{theorem}\label{th:recovery-net}
The out-asymptotic net $\A^\tout_+\otimes \A^\tout_-$ and the S-matrix $S$ completely
characterizes the original net $\A$ if it satisfies Bisognano-Wichmann
property, Haag duality and asymptotic completeness.
\end{theorem}
\begin{proof}
The wedge algebra is recovered by
$\M = \A(W_\R) = \{\Phi^\tout_+(x),\Phi^\tin_-(y): x,y\in \M\}''$
by Proposition \ref{pr:recovery-wedge}. In the right-hand side, $\Phi^\tin_-$ is recovered
 from $\Phi^\tout_-$ and $S$ by Proposition \ref{pr:recovery-field}. Hence the wedge algebra
is completely recovered from the data $\Phi^\tout_\pm$ and $S$, or $\A^\tin_\pm$ and $S$
by Proposition \ref{pr:recovery-wedge}. By Haag duality,
the data of wedge algebras are enough to recover the local algebras. By Bisognano-Wichmann
property, the representation $U$ of the whole Poincar\'e group is recovered from
the modular data.
\end{proof}
\begin{remark}
Among the conditions on $\A$, Bisognano-Wichmann property is satisfied in almost all known
examples. Haag duality can be satisfied by extending the net \cite{Baumgaertel} without
changing the S-matrix. Hence we consider them as standard assumptions. On the other
hand, asymptotic completeness is in fact a very strong condition. For example,
a conformal net is asymptotically complete if and only if it is chiral \cite{Tanimoto11-2}.
Hence the class of asymptotically complete nets could be very small even among Poincar\'e
covariant nets. But a clear-cut scattering theory is available only for asymptotically
complete cases. The general case is under investigation \cite{DT11-3}.
\end{remark}

\subsection{Chiral nets as asymptotic nets}\label{recovery}

We can express the modular objects of the interacting net in terms of the ones of the asymptotic chiral net.
\begin{proposition}\label{pr:modular-objects}
Let $\Delta^\tout$ and $J^\tout$ be the modular operator and the modular conjugation
of $\A_+^\tout(\RR_-)\otimes\A_-^\tout(\RR_+)$ with respect to $\Omega$. Then it holds
that $\Delta = \Delta^\tout$ and $J = SJ^\tout$.
\end{proposition}
\begin{proof}
First we note that the modular objects of $\A(W_\R)$ restrict to
$\H_+$ and $\H_-$ by Takesaki's theorem \cite[Theorem IX.4.2]{TakesakiII}.
Indeed, $\A^\tout_+(\RR_+)$ and $\A^\tout_-(\RR_-)$ are subalgebras of
$\A(W_\R)$ and invariant under $\Ad \Delta^{it}$, or equivalently under
the Lorentz boosts $\Ad U(\Lambda(-2\pi t))$ by Bisognano-Wichmann property, as
we saw in the proof of Proposition \ref{pr:recovery-wedge}, then the
projections onto the respective subspaces commute with the modular objects.
Let us denote these restrictions by $\Delta_+^{it},J_+,\Delta_-^{it}$ and
$J_-$, respectively.

We identify $\H_+\otimes\H_-$ and the full Hilbert space $\H$ by the action of
$\A^\tout_+\otimes\A^\tout_-$.
By Bisognano-Wichmann property and Lemma \ref{lm:collision-states}, we have
\begin{eqnarray*}
\Delta^{it}\cdot \xi\timeso\eta
&=& (U(\Lambda(-2\pi t))\xi)\timeso (U(\Lambda(-2\pi t))\eta)\\
&=& \Delta_+^{it}\xi\otimes \Delta_-^{it}\eta\\
&=& (\Delta_+\otimes\Delta_-)^{it} \cdot \xi\otimes\eta,
\end{eqnarray*}
which implies that $\Delta = \Delta_+\otimes \Delta_- = \Delta^\tout$.

As for modular conjugations, we take $x \in \A(W_\R)$ and $y \in \A(W_\R)' = \A(W_\L)$
and set $\xi = \Phi^\tout_+(x)\Omega$ and $\eta = \Phi^\tout_-(y)\Omega$.
Then we use Lemma \ref{lm:asymptotic-field2} to see
\begin{eqnarray*}
J\cdot \xi\timeso\eta
&=& J\cdot \Phi^\tout_+(x)\Phi^\tout_-(y)\Omega\\
&=& \Phi^\tin_+(JxJ)\Phi^\tin_-(JyJ)\Omega\\
&=& (J\xi)\timesi(J\eta)\\
&=& S\cdot (J_+\xi)\timeso(J_-\eta)\\
&=& S\cdot (J_+\otimes J_-) \cdot \xi\otimes\eta
\end{eqnarray*}
from which one infers that $J = S\cdot (J_+\otimes J_-) = S\cdot J^\tout$.
\end{proof}

Theorem \ref{th:recovery-net} tells us that chiral conformal nets can be viewed as
free field nets for massless two-dimensional theory (cf. \cite{Tanimoto11-2}).
Let us formulate the situation the other way around. Let $\A_+\otimes \A_-$ be a chiral CFT,
then it is an interesting open problem to characterize unitary operators which
can be interpreted as a S-matrix of a net whose asymptotic net is the given $\A_+\otimes \A_-$.
We restrict ourselves to point out that
there are several immediate necessary conditions:
For example, $S$ must commute with the Poincar\'e symmetry of the chiral net
since it coincides with the one of the interacting net. Analogously it must hold that
$(J_+\otimes J_-) S (J_+\otimes J_-) = S^*$. Furthermore, the algebra of the form as
in Proposition \ref{pr:recovery-wedge} must be strictly local.

If one has an appropriate operator $S$, an interacting Borchers triple can be constructed by
(cf. Propositions \ref{pr:recovery-wedge}, \ref{pr:recovery-field})
\begin{itemize}
\item $\M_S := \{x\otimes\1, \Ad S(\1\otimes y): x\in\A_+(\RR_-), y\in\A_-(\RR_+)\}''$,
\item $U:= U_+\otimes U_-$,
\item $\Omega := \Omega_+\otimes \Omega_-$.
\end{itemize}

By the formula for the modular conjugation in Proposition \ref{pr:modular-objects},
it is immediate to see that
\[
\M_S' := \{\Ad S(x'\otimes\1), \1\otimes y': x'\in\A_+(\RR_+), y'\in\A_-(\RR_-)\}''.
\]
Then for $x\in\A_+(\RR_-), y\in\A_-(\RR_+)$ it holds that
$\Phi^\tout_+(x\otimes \1) = x\otimes\1$ and
$\Phi^\tin_-(\Ad S(\1\otimes y)) = \Ad S(\1\otimes y)$.
Similarly,
we have $\Phi^\tout_-(\Ad S(x'\otimes \1)) = \Ad S(x'\otimes \1)$ and $\Phi^\tin_+(\1\otimes y') = \1\otimes y'$
for $x' \in \A_+(\RR_+)$ and $y' \in \A_-(\RR_-)$. From this it is easy to see that
$S$ is indeed the S-matrix of the constructed Borchers triple.

In the following Sections we will construct unitary operators
which comply with these conditions except strict locality.
To my opinion, however, the true difficulty is the strict locality,
which has been so far established only for ``regular'' massive models \cite{Lechner08}.
But it is also true that the class of S-matrices constructed in the present paper
can be seen rather small (see the discussion in Section \ref{conclusion}).

\section{Construction through one-parameter semigroup of endomorphisms}\label{one-parameter}
In this Section, we construct families of Borchers triples using one-parameter semigroup
of endomorphisms of Longo-Witten type. The formula to define
the von Neumann algebra is very simple and the proofs use a common argument based
on spectral decomposition.

Our construction is based on chiral conformal nets on $S^1$, and indeed one
family can be identified as the BLS deformation of chiral nets
(see Section \ref{deformation-translation}). But in our construction, the meaning
of the term ``deformation'' is not clear and we refrain from using it.
 From now on, we consider only chiral net with the identical components $\A_1 = \A_2 = \A_0$
for simplicity. It is not difficult to generalize it to ``heterotic case'' where $\A_1 \neq \A_2$.

\subsection{The commutativity lemma}\label{commutativity}
The following Lemma is the key of all the arguments and will be used later
in this Section concerning one-parameter endomorphisms.
Typical examples of the operator $Q_0$ in Lemma will be the generator of
one-dimensional translations $P_0$ (Section \ref{deformation-translation}),
or of one-parameter inner symmetries of the chiral component (Section \ref{deformation-inner}).

As a preliminary, we give a remark on tensor product. See \cite{Dixmier81}
for a general account on spectral measure and measurable family.
Let $E_0$ be a projection-valued measure on $Z$ (typically, the spectral measure
of some self-adjoint operator) and $V(\lambda)$ be a measurable family of
operators (bounded or not). Then one can define an operator
\[
\int_Z V(\lambda)\otimes dE_0(\lambda) (\xi\otimes\eta)
:= \int_Z V(\lambda)\xi\otimes dE_0(\lambda)\eta.
\]
If $V(\lambda)$ is unbounded, the vector $\xi$ should be in a common domain
of $\{V(\lambda)\}$. As we will see, this will not matter in our cases.
For two bounded measurable families $V, V'$, it is easy to see that
\[
\int_Z V(\lambda)\otimes dE_0(\lambda) \cdot \int_Z V'(\lambda)\otimes dE_0(\lambda)
= \int_Z V(\lambda)V'(\lambda)\otimes dE_0(\lambda).
\]

\begin{lemma}\label{lm:commutativity}
We fix a parameter $\k \in \RR$. Let $Q_0$ be a self-adjoint operator on $\H_0$ and
Let $Z \subset \RR$ be the spectral supports of $Q_0$. If it holds that
$[x,\Ad e^{is\k Q_0}(x')] = 0$ for $x,x'\in B(\H_0)$ and $s \in Z$, then
we have that
\begin{align*}
[x\otimes \1,\Ad e^{i\k Q_0\otimes Q_0}(x'\otimes \1)] &= 0,\\
[\1\otimes x,\Ad e^{i\k Q_0\otimes Q_0}(\1\otimes x')] &= 0.
\end{align*}
\end{lemma}
\begin{proof}
We prove only the first commutation relation, since the other is analogous.
Let $Q_0 = \int_Z s\cdot dE_0(s)$ be the spectral decomposition of $Q_0$.
According to this spectral decomposition, we can decompose only the second component:
\[
Q_0\otimes Q_0 = Q_0\otimes \int_Z s\cdot dE_0(x) = \int_Z sQ_0\otimes dE_0(s).
\]
Hence we can describe the adjoint action of $e^{i\k Q_0\otimes Q_0}$ explicitly:
\begin{eqnarray*}
\Ad e^{i\k Q_0\otimes Q_0}(x'\otimes \1)
&=& \int_Z e^{is\k Q_0}\otimes dE(s) \cdot (x'\otimes \1)\cdot \int_Z e^{-is\k Q_0}\otimes dE_0(s)\\
&=& \int_Z \left(\Ad e^{is\k Q_0} (x') \right)\otimes dE_0(s)
\end{eqnarray*}
Then it is easy to see that this commutes with $x\otimes \1$ by
the assumed commutativity.
\end{proof}

\subsection{Construction of Borchers triples with respect to
translation}\label{deformation-translation}
The objective here is to apply the commutativity lemma in Section \ref{commutativity}
to the endomorphism of translation.
Then it turns out that the Borchers triples obtained by the BLS deformation
of a chiral triple coincide with this construction. A new feature is that our
construction involves only von Neumann algebras.

\subsubsection{Construction of Borchers triples}\label{borchers-translation}
Let $(\M,T,\Omega)$ be a chiral Borchers triple with chiral component
$\A_0$ and $T_0(t) = e^{itP_0}$ the chiral translation:
Namely,
$\M = \A_0(\RR_-)\otimes \A_0(\RR_+)$,
$T(t_0,t_1) = T_0\left(\frac{t_0-t_1}{\sqrt{2}}\right)\otimes T_0\left(\frac{t_0+t_1}{\sqrt{2}}\right)$
and $\Omega = \Omega_0\otimes\Omega_0$.

Note that $T_0(t)$ implements a Longo-Witten endomorphism of $\A_0$
for $t \ge 0$. In this sense, the construction of this Section is considered
to be based on the endomorphisms $\{\Ad T_0(t)\}$. A nontrivial family of endomorphisms
will be featured in Section \ref{u1-current}.

We construct a new Borchers triple on the same Hilbert space
$\H = \H_0\otimes\H_0$ as follows.
Let us fix $\k \in \RR_+$.
\begin{itemize}
\item $\M_{P_0,\k} := \{x\otimes \1, \Ad e^{i\k P_0\otimes P_0}(\1\otimes y)
: x\in \A_0(\RR_-), y\in\A_0(\RR_+)\}''$,
\item the same $T$ from the chiral triple,
\item the same $\Omega$ from the chiral triple.
\end{itemize}

\begin{theorem}\label{th:triple-translation}
Let $\k \ge 0$. Then the triple $(\M_{P_0,\k},T,\Omega)$ is a Borchers triple with
the S-matrix $S_{P_0,\k} = e^{i\k P_0\otimes P_0}$.
\end{theorem}
\begin{proof}
The vector $\Omega_0\otimes\Omega_0$ is obviously invariant under $T$ and $T$
has the spectrum contained in $V_+$.
The generator $P_0$ of one-dimensional translations obviously commutes with one-dimensional
translation $T_0$, hence $P_0\otimes P_0$ commutes with $T = T_0\otimes T_0$,
so does $e^{i\k P_0\otimes P_0}$.
We claim that $\M_{P_0,\k}$ is preserved under translations in the right wedge.
Indeed, if $(t_0,t_1) \in W_\R$, then we have
\[
\Ad T(t_0,t_1) \left(\Ad e^{i\k P_0\otimes P_0}(\1\otimes y)\right)
= \Ad e^{i\k P_0\otimes P_0}\left(\Ad T(t_0,t_1) (\1\otimes y)\right).
\]
and $\Ad T(t_0,t_1)(\1\otimes y) \in \1\otimes \A_0(\RR_+)$ and it is obvious
that $\Ad T(t_0,t_1)(x\otimes \1) \in \A_0(\RR_-)\otimes\1$, hence
the generators of the von Neumann algebra $\M_{P_0,\k}$ are preserved.

We have to show that $\Omega$ is cyclic and separating for $\M_{P_0,\k}$.
Note that it holds that $e^{i\k P_0\otimes P_0}\cdot \xi\otimes\Omega_0 = \xi\otimes\Omega_0$
for any $\k \in \RR, \xi \in \H_0$, by the spectral calculus.
Now cyclicity is seen by noting that
\begin{eqnarray*}
(x\otimes \1)\cdot \Ad e^{i\k P_0\otimes P_0} (\1\otimes y) \cdot \Omega
&=& (x\otimes \1)\cdot e^{i\k P_0\otimes P_0} \cdot (x\Omega_0) \otimes \Omega_0\\
&=& (x\Omega_0) \otimes (y\Omega_0)
\end{eqnarray*}
and by the cyclicity of $\Omega$ for the original algebra $\M = \A_0(\RR_-)\otimes \A_0(\RR_+)$.

Finally we show separating property as follows: we set
\begin{equation*}
\M^1_{P_0,\k} = \{\Ad e^{i\k P_0\otimes P_0}(x'\otimes \1),
\1\otimes y': x'\in \A_0(\RR_+), y'\in\A_0(\RR_-)\}''.
\end{equation*}
Note that $\Omega$ is cyclic for $\M^1_{P_0,\k}$ by an analogous proof for $\M_{P_0,\k}$,
thus for the separating property, it suffices to show that
$\M_{P_0,\k}$ and $\M^1_{P_0,\k}$ commute.
Let $x,y' \in \A_0(\RR_-), x' \in \A_0(\RR_+)$. First,
$x\otimes \1$ and $\1\otimes y'$ obviously commute.
Next, we apply Lemma \ref{lm:commutativity} to $x, x'$ and $Q_0 = P_0$ to see that
$x\otimes \1$ and $\Ad e^{i\k P_0\otimes P_0}(x'\otimes \1)$ commute:
Indeed, the spectral support of $P_0$ is $\RR_+$, and for $s \in \RR_+$,
$x$ and $\Ad e^{is\k P_0}(x')$ commute since $P_0$ is the generator of
one-dimensional translations and since $x\in\A_0(\RR_-), x'\in\A_0(\RR_+)$.
Similarly, for $y \in \A_0(\RR_+)$, $\Ad e^{i\k P_0\otimes P_0}(\1\otimes y)$ and $\M^1_{P_0,\k}$ commute.
This implies that $\M_{P_0,\k}$ and $\M^1_{P_0,\k}$ commute.

The S-matrix corresponds to the unitary used to twist the chiral net as we saw
in the discussion at the end of Section \ref{recovery}.
\end{proof}

Now that we have constructed a Borchers triple, it is possible to express its modular
objects in terms of the ones of the chiral triple by an analogous argument
as Proposition \ref{pr:modular-objects}. Then one sees that $\M^1_{P_0,\k}$ is
indeed the commutant $\M_{P_0,\k}'$.

\subsubsection{BLS deformation}\label{bls-deformation}
We briefly review the BLS deformation \cite{BLS}. 
Let $(\M, T, \Omega)$ be a Borchers triple. We denote by $\M^\infty$ the subset of
elements of $\M$ which are smooth under the action of $\a = \Ad T$ in the norm topology.
Then one can define for any $x\in\M^{\infty}$,
and a matrix $\Theta_\k = \left(\begin{matrix} 0 & \k \\ \k & 0 \end{matrix}\right)$,
the warped convolution
\begin{equation*}
x_\k=\int dE(a)\,\a_{\Theta_\k a}(x)
:=\lim_{\e\searrow 0}(2\pi)^{-2}\int d^2a\, d^2b\, f(\e a,\e b) e^{-ia\cdot b}\a_{\Theta_\k a}(x)T(b) \label{warped}
\end{equation*}
on a suitable domain,
where $dE$ is the spectral measure of $T$ and $f\in \mathscr{S}(\RR^2\times\RR^2)$ satisfies $f(0,0)=1$.
We set
\begin{equation*}
\M_\k:=\{x_\k: x\in\M^{\infty}\}''.
\end{equation*}
For $\k>0$, the following holds.
\begin{theorem}[\cite{BLS}]\label{th:BLS}
If $(\M, T, \Omega)$ is a Borchers triple, then 
$(\M_\k, T, \Omega)$ is also a Borchers triple.
\end{theorem}
We call the latter the {\bf BLS deformation}
of the original triple $(\M,T,\Omega)$. One of the main results of this paper is to obtain
the BLS deformation by a simple procedure.

We have determined the property of deformed scattering theory in \cite{DT11-1}.
In our notation $M^2 = P_0\otimes P_0$ we have the following.
\begin{theorem}\label{th:states-deformations}
For any $\xi \in \H_+$ and $\eta \in \H_-$, the following relations hold:
\begin{eqnarray*}
 & &\xi\timeso_\k \eta
 =e^{-\frac{i\k}{2}P_0\otimes P_0}(\xi\otimes \eta),\\
 & &\xi\timesi_\k \eta
 =e^{\frac{i\k}{2}P_0\otimes P_0}(\xi\otimes \eta),
\end{eqnarray*}
where on the left-hand sides there appear the collision states of the deformed theory. 
\end{theorem}

\subsubsection{Reproduction of BLS deformation}\label{reproduction-bls}
Let $(\M,T,\Omega)$ be a chiral Borchers triple.
In this Section we show that the Borchers triple $(\M_{P_0,\k},T,\Omega)$ obtained above
is unitarily equivalent to the BLS deformation $(\M_{\k},T,\Omega)$. Then we can calculate the
asymptotic fields very simply. We use symbols $\timeso$ and $\timeso_{\k}$
to denote collision states with respect to the corresponding Borchers triples with
$\M$ (undeformed) and $\M_{\k}$, respectively.
Recall that for the undeformed chiral triple, all these products $\timeso$, $\timesi$ and
$\otimes$ coincide \cite{DT11-1}.
\begin{theorem}\label{th:coincidence}
Let us put $\N_{P_0,\k} := \Ad e^{-\frac{i\k}{2}P_0\otimes P_0} \M_{P_0,\k}$. 
Then it holds that $\N_{P_0,\k} = \M_{\k}$, hence we have the coincidence
of two Borchers triples
$(\N_{P_0,\k},T,\Omega) = (\M_{\k},T,\Omega)$.
\end{theorem}
\begin{proof}
In \cite{DT11-1}, we have seen that
the deformed BLS triple is asymptotically complete. Furthermore, we have
\[
\xi \timeso_{\k} \eta = e^{-\frac{i\k}{2}P_0\otimes P_0}\xi \timeso \eta.
\]
As for observables, let $x \in A_0(\RR_-)$ and we use the notation $x_{\Theta_\k}$ from \cite{BLS}
\footnote{The reader is suggested to look at the notation $F_Q$ in \cite{BLS},
where $F$ is an observable in $\M$ and $Q$ is a $2\times 2$ matrix.
We keep the symbol $Q$ for a generator of one-parameter automorphisms, hence
we changed the notation to avoid confusions.}.
For the asymptotic field $\Phi^\tout_{\k,+}$ of BLS deformation, we have
\begin{eqnarray*}
\Phi_{\k,+}^\tout((x\otimes \1)_{\Theta_{\k}}) \xi\timeso_{\k} \eta
&=& ((x\otimes \1)_{\Theta_{\k}}\xi)\timeso_{\k} \eta \\
&=& (x\xi)\timeso_{\k} \eta \\
&=& e^{-\frac{i\k}{2}P_0\otimes P_0} \cdot (x\xi)\timeso \eta \\
&=& e^{-\frac{i\k}{2} P_0\otimes P_0} \cdot x\otimes \1 \cdot \xi\timeso \eta \\
&=& \Ad e^{-\frac{i\k}{2} P_0\otimes P_0} (x\otimes \1) \cdot e^{-\frac{i\k}{2} P_0\otimes P_0} \cdot \xi\timeso \eta \\
&=& \Ad e^{-\frac{i\k}{2} P_0\otimes P_0} (x\otimes \1) \cdot \xi\timeso_{\k} \eta,
\end{eqnarray*}
(see \ref{a-remark} for the second equality)
hence, we have $\Phi_{\k,+}^\tout((x\otimes\1)_{\Theta_{\k}})
= \Ad e^{-\frac{i\k}{2} P_0\otimes P_0} (x\otimes \1)$.
Analogously we have $\Phi_{\k,-}^\tin((\1\otimes y)_{\Theta_\k})
= \Ad e^{\frac{i\k}{2} P_0\otimes P_0} (\1\otimes y)$ for
$y \in \A_0(\RR_+)$.

Note that by definition we have
\[
\N_{P_0,\k} = \{\Ad e^{-\frac{i\k}{2}P_0\otimes P_0} (x\otimes \1),
\Ad e^{\frac{i\k}{2}P_0\otimes P_0}(\1\otimes y): x\in\A_0(\RR_-),y\in\A_0(\RR_+)\}''.
\]
Since the image of the right-wedge algebra by $\Phi^\tout_+$ and
$\Phi^\tin_-$ remains in the right-wedge algebra, from the above observation,
we see that $\N_{P_0,\k} \subset \M_{\k}$ \cite{DT11-1}.
To see the converse inclusion, recall that it has been proved that
the modular group $\Delta^{it}$ of the right-wedge algebra with respect to
$\Omega$ remains unchanged under the BLS deformation.
We have that $\Ad \Delta^{it} (e^{i\k P_0\otimes P_0}) = e^{i\k P_0\otimes P_0}$,
hence it is easy to see that $\N_{P_0,\k}$ is invariant under $\Ad \Delta^{it}$.
By the theorem of Takesaki \cite[Theorem IX.4.2]{TakesakiII},
there is a conditional expectation from $\M_{\k}$ onto $\N_{P_0,\k}$ which preserves
the state $\<\Omega,\cdot\Omega\>$ and in particular,
$\M_{\k} = \N_{P_0,\k}$ if and only if $\Omega$ is cyclic for $\N_{P_0,\k}$.
We have already seen the cyclicity in Theorem \ref{th:triple-translation},
thus we obtain the thesis.

The translation $T$ and $\Omega$ remain unchanged under $e^{-\frac{i\k}{2}P_0\otimes P_0}$,
which established the unitary equivalence between two Borchers triples.
\end{proof}

\begin{remark}
It is also possible to formulate Theorem \ref{th:recovery-net} for Borchers triple,
although the asymptotic algebra will be neither local nor conformal in general.
 From this point of view, Theorem \ref{th:coincidence} is just a corollary of
the coincidence of S-matrix. Here we preferred a direct proof, instead of
formulating non local net on $\RR$.
\end{remark}

\subsection{Endomorphisms with asymmetric spectrum}
Here we briefly describe a generalization of the construction in previous Sections.
Let $\A_0$ be a local net on $S^1$, $T_0$ be the representation of the translation.
We assume that there is a one-parameter family $V_0(t) = e^{iQ_0 t}$ of unitary operators with a
positive or negative generator $Q_0$ such that $V_0(t)$ and $T_0(s)$ commute and
$\Ad V_0(t)(\A_0(\RR_+)) \subset \A_0(\RR_+)$ for $t \ge 0$.
With these ingredients, we have the following:
\begin{theorem}
The triple
\begin{itemize}
\item $\M_{Q_0,\k} := \{x\otimes \1, \Ad e^{\pm i\k Q_0\otimes Q_0} (\1\otimes y): x\in\A_0(\RR_-), y\in\A_0(\RR_+)\}''$,
\item $T := T_0\otimes T_0$,
\item $\Omega := \Omega_0\otimes \Omega_0$,
\end{itemize}
where $\pm$ corresponds to $\sp Q_0 \subset \RR_\pm$, is a Borchers triple with the S-matrix
$e^{\pm i\k Q_0\otimes Q_0}$ for $\k \ge 0$.
\end{theorem}
The proof is analogous to Theorem \ref{th:triple-translation} and we refrain from repeating it here.

The construction looks very simple, but to our knowledge, there are only few examples.
The one-parameter group of translation itself has been studied in the previous Sections.
Another one-parameter family of unitaries with a negative generator
$\{\G(e^{-{\frac{\k}{P_1}}})\}$ has been found for
the $U(1)$-current \cite{LW11}, where
$P_1$ is the generator on the one-particle space, $\k \ge 0$ and
$\G$ denotes the second quantization.
Indeed, by Borchers' theorem \cite{Borchers92, Florig98}, such
one-parameter group together with the modular group forms a representation of
the ``$ax+b$'' group, thus it is related somehow with translation.

\subsection{Construction of Borchers triples through inner symmetry in chiral CFT}\label{deformation-inner}
\subsubsection{Inner symmetry}\label{inner-symmetry}
Let $\A_0$ be a conformal (M\"obius) net on $S^1$. An {\bf automorphism} of $\A_0$
is a family of automorphisms $\{\a_{0,I}\}$ of local algebras $\{\A_0(I)\}$ with
the consistency condition $\a_{0,J}|_{\A_0(I)} = \a_{0,I}$ for $I \subset J$.
If each $\a_{0,I}$ preserves the vacuum state $\omega$, then $\a_0$ is said to be
an {\bf inner symmetry}. An inner symmetry $\a_0$ is implemented by a unitary $V_{\a_0}$
defined by $V_{\a_0} x\Omega = \a_{0,I}(x)\Omega$, where $x \in \A_0(I)$. This definition
does not depend on the choice of $I$ by the consistency condition. If $\a_{0,t}$ is
a one-parameter family of weakly continuous automorphisms, then the implementing unitaries satisfy
$V_{\a_0}(t) V_{\a_0}(s) = V_{\a_0}(t+s)$ and $V_{\a_0}(0) = \1$, hence there is a self-adjoint
operator $Q_0$ such that $V_{\a_0}(t) = e^{itQ_0}$ and $Q_0\Omega = 0$. Furthermore,
$e^{itQ_0}$ commutes with modular objects \cite{TakesakiII}: $J_0e^{itQ_0}J_0 = e^{itQ_0}$, or
$J_0Q_0J_0 = -Q_0$ (note that $J_0$ is an anti-unitary involution). If $\a_t$ is periodic
with period $2\pi$, namely $a_{0,t} = a_{0,t+2\pi}$ then it holds that $V_{\a_0}(t) = V_{\a_0}(t+2\pi)$
and the generator $Q_0$ has a discrete spectrum $\sp Q_0 \subset \ZZ$.
For the technical simplicity, we restrict ourselves to the study of periodic inner symmetries.
We may assume that the period is $2\pi$ by a rescaling of the parameter.

\begin{example}
We consider the loop group net $\A_{G,k}$ of a (simple, simply connected)
compact Lie group $G$ at level $k$ \cite{GF93, Wassermann}, the net generated by vacuum representations
of loop groups $LG$ \cite{PS86}. On this net, the original group $G$ acts as a group of
inner symmetries. We fix a maximal torus in $G$ and choose a one-parameter group in the
maximal torus with a rational direction, then it is periodic. Any one-parameter group
is contained in a maximal torus, so there are a good proportion of periodic one-parameter
groups in $G$ (although generic one-parameter groups have irrational direction, hence not
periodic). In particular, in the $SU(2)$-loop group net $\A_{SU(2),k}$, any one-parameter
group in $SU(2)$ is periodic since $SU(2)$ has rank $1$.
\end{example}

An inner automorphism $\a_0$ commutes with M\"obius symmetry because of Bisognano-Wichmann
property. Hence it holds that $U_0(g)Q_0U_0(g)^* = Q_0$.
Furthermore, if the net $\A_0$ is conformal, then $\a_0$ commutes also with
the diffeomorphism symmetry \cite{CW05}.
Let $G$ be a group of inner symmetries and $\A_0^G$ be the assignment: $I\mapsto \A_0(I)^G|_{\H_0^G}$,
where $\A_0(I)^G$ denotes the fixed point algebra of $\A_0(I)$ with respect to $G$
and $\H_0^G := \overline{\{x\Omega_0: x \in \A_0^G(I), I \subset S^1\}}$.
Then it is easy to see that $\A_0^G$ is a M\"obius covariant net and it is referred to
as the {\bf fixed point subnet} of $\A_0$ with respect to $G$.

We can describe the action $\a_0$ of a periodic one-parameter group of inner symmetries
in a very explicit way, which can be considered as the ``spectral decomposition'' of $\a_0$.
Although it is well-known, we summarize it here with a proof for the later use.
This will be the basis of the subsequent analysis.
\begin{proposition}\label{pr:spectral-decomposition}
Any element $x \in \A_0(I)$ can be written as $x = \sum_n x_n$, where
$x_n \in \A_0(I)$ and $\a_{0,t}(x_n) = e^{int}x_n$.
We denote $\A_0(I)_n = \{x \in \A_0(I): \a_{0,t}(x) = e^{int}x\}$.
It holds that $\A_0(I)_m\A_0(I)_n \subset \A_0(I)_{m+n}$ and
$\A_0(I)_m E_0(n)\H_0 \subset E_0(m+n)\H_0$, where $E_0(n)$ denotes the spectral
projection of $Q_0$ corresponding to the eigenvalue $n \in \ZZ$.
\end{proposition}
\begin{proof}
Let us fix an element $x \in \A_0(I)$.
The Fourier transform
\[
x_n := \int_0^{2\pi} \a_s(x)e^{-ins}\, ds
\]
(here we consider the weak integral using the local normality of $\a_{0,t}$)
is again an element of $\A_0(I)$, since $\A_0(I)$ is invariant under
$\a_{0,t}$. Furthermore it is easy to see that
\begin{eqnarray*}
\a_{0,t}(x_n)
&=& \a_{0,t}\left(\int_0^{2\pi} \a_{0,s}(x)e^{-ins}\, ds\right) \\
&=& \int_0^{2\pi} \a_{0,s+t}(x)e^{-ins}\, ds \\
&=& e^{int}\int_0^{2\pi} \a_{0,s}(x)e^{-ins}\, ds = e^{int}x_n,
\end{eqnarray*}
hence we have $x_n \in \A_0(I)_n$.

By assumption, $\a_{0,t}(x) = \Ad e^{itQ_0}(x)$ and $\sp Q_0 \subset \ZZ$.
If we define $x_{l,m} = E_0(l)xE_0(m)$, it holds that $\Ad e^{itQ_0} x_{l,m} = e^{i(l-m)t}x_{l,m}$.
The integral and this decomposition into matrix elements are compatible, hence
for $x \in \A_0(I)$ we have
\[
x_n = \sum_{l-m = n} x_{l,m}.
\]
Now it is clear that $x = \sum_n x_n$ where  each summand is a different matrix element,
hence the sum is strongly convergent. Furthermore from this decomposition we see that
$\A_0(I)_m\A_0(I)_n \subset \A_0(I)_{m+n}$ and $\A_0(I)_m E_0(n)\H_0 \subset E_0(m+n)\H_0$.
\end{proof}

At the end of this Section,
we exhibit a simple formula for the adjoint action $\Ad e^{i\k Q_0\otimes Q_0}$
on the tensor product Hilbert space $\H := \H_0\otimes \H_0$.
\begin{lemma}\label{lm:adjoint-action}
For $x_m \in \A_0(I)_m, y_n \in \A_0(I)_n$,
it holds that $\Ad e^{i\k Q_0\otimes Q_0}(x_m\otimes \1) = x_m\otimes e^{im\k Q_0}$
and $\Ad e^{i\k Q_0\otimes Q_0}(\1\otimes y_n) = e^{in\k Q_0}\otimes y_n$.
\end{lemma}
\begin{proof}
Recall that $\sp Q_0 \in \ZZ$. Let $Q_0 = \sum_l l\cdot E_0(l)$ be the spectral decomposition of $Q_0$.
As in the proof of Lemma \ref{lm:commutativity}, we decompose only the second component of
$Q_0\otimes Q_0$ to see that
\begin{eqnarray*}
Q_0\otimes Q_0 &=& Q_0\otimes \left(\sum_l l\cdot E_0(l)\right) = \sum_l lQ_0\otimes E_0(l) \\
e^{i\k Q_0\otimes Q_0} &=& \sum_l e^{il\k Q_0}\otimes E_0(l) \\
\Ad e^{i\k Q_0\otimes Q_0}(x_m\otimes \1) &=& \sum_l \Ad e^{il\k Q_0}(x_m)\otimes E_0(l)\\
&=& \sum_l e^{iml\k}x_m\otimes E_0(l)\\
&=& x_m\otimes e^{im\k Q_0}.
\end{eqnarray*}
\end{proof}

\begin{proposition}
For each $l \in \ZZ$ there is a cyclic and separating vector $v \in E_0(l)\H_0$ for
a local algebra $\A_0(I)$.
\end{proposition}
\begin{proof}
It is enough to note that the decomposition $\1 = \sum_lE_0(l)$ is compatible
with the decomposition of the whole space with respect to rotations, since
inner symmetries commute with any M\"obius transformation.
Hence each space $E_0(l)\H_0$ is a direct sum of eigenspace of rotation.
It is a standard fact that a eigenvector of rotation which has positive
spectrum is cyclic and separating for each local algebra (see the standard proof
of Reeh-Schlieder property, e.g. \cite{Baumgaertel}).
\end{proof}
We put $E(l,l') := E_0(l)\otimes E_0(l')$.
\begin{corollary}\label{cr:separating}
Each space $E(l,l')\H$ contains a cyclic and separating vector $v$ for
$\A_0(I)\otimes \A_0(J)$ for any pair of intervals $I,J$.
\end{corollary}

\subsubsection{Construction of Borchers triples and their intersection property}\label{borchers-inner}
Let $\A_0$ be a M\"obius covariant net and $\a_{0,t}$ be a periodic one-parameter group of
inner symmetries. The automorphisms can be implemented as $\a_{0,t} = \Ad e^{itQ_0}$
as explained in Section \ref{inner-symmetry}. The self-adjoint operator $Q_0$ is
referred to as the generator of the inner symmetry.

We construct a Borchers triple as in Section \ref{borchers-translation}.
Let $\k \in \RR$ be a real parameter (this time $\k$ can be positive or negative)
and we put
\begin{eqnarray*}
\M_{Q_0,\k} &:=& \{x\otimes \1, \Ad e^{i\k Q_0\otimes Q_0}(\1\otimes y):
 x \in \A_0(\RR_-), y \in \A_0(\RR_+)\}'' \\
T(t_0,t_1) &:=& T_0\left(\frac{t_0-t_1}{\sqrt{2}}\right)\otimes T_0\left(\frac{t_0+t_1}{\sqrt{2}}\right) \\
\Omega &:=& \Omega_0\otimes \Omega_0
\end{eqnarray*}

\begin{theorem}
The triple $(\M_{Q_0,\k},T,\Omega)$ above is a Borchers triple with a nontrivial scattering
operator $S_{Q_0,\k} = e^{i\k Q_0\otimes Q_0}$.
\end{theorem}
\begin{proof}
As remarked in Section \ref{inner-symmetry}, $Q_0$ commutes with M\"obius symmetry $U_0$,
hence $Q_0\otimes Q_0$ and the translation $T = T_0\otimes T_0$ commute.
Since $(\A_0(\RR_-)\otimes\A_0(\RR_+), T,\Omega)$
is a Borchers triple (see Section \ref{preliminaries}), it holds that
$\Ad T(t_0,t_1) \M \subset \M$ for $(t_0,t_1) \in W_\R$ and $T(t_0,t_1)\Omega = \Omega$
and $T$ has the joint spectrum contained in $V_+$.

Since $\a_{0,t}$ is a one-parameter group of inner symmetries, it holds that
$\a_{0,s}(\A_0(\RR_-)) = \A_0(\RR_-)$ and $\a_{0,t}(\A_0(\RR_+)) = \A_0(\RR_+)$ for $s,t \in \RR$.
By Lemma \ref{lm:commutativity}, for $x \in \A_0(\RR_-)$ and $x' \in \A_0(\RR_+)$ it holds that
\[
[x\otimes \1, \Ad e^{i\k Q_0\otimes Q_0}(x'\otimes \1)] = 0.
\]
Then one can show that $(\M_{Q_0,\k},T,\Omega)$ is a Borchers triple as in the proof of
Theorem \ref{th:triple-translation}. The formula for the S-matrix can be proved
analogously as in Section \ref{recovery}.
\end{proof}

We now proceed to completely determine the intersection property of $\M_{Q_0,\k}$.
As a preliminary, we describe the elements in $\M_{Q_0,\k}$ in terms of
the original algebra $\M$ componentwise. On $\M = \A_0(\RR_-)\otimes \A_0(\RR_+)$,
there acts the group $S^1\otimes S^1$ by the tensor product action:
$(s,t)\mapsto \a_{s,t} := \a_{0,s}\otimes\a_{0,t} = \Ad (e^{isQ_0}\otimes e^{itQ_0})$.
According to this action,
we have a decomposition of an element $z \in \M$ into Fourier components
as in Section \ref{inner-symmetry}:
\[
z_{m,n} := \int_{S^1\times S^1} \a_{s,t}(z)e^{-i(ms+nt)}\,ds\,dt,
\]
which is still an element of $\M$,
and with $E(l,l') := E_0(l)\otimes E_0(l')$, these components can be obtained by
\[
z_{m,n} = \underset{l'-k' = n}{\sum_{l-k = m}}E(l,l')z E(k,k').
\]

One sees that $\Ad (e^{isQ_0}\otimes e^{itQ_0})$ acts also on $\M_{Q_0,\k}$ since
it commutes with $\Ad e^{i\k Q_0\otimes Q_0}$. We still write this action by $\a$.
We can take their Fourier components by the same formula and the formula
with spectral projections still holds.

\begin{lemma}\label{lm:component}
An element $z_\k\in \M_{Q_0,\k}$ has the components of the form
\[
(z_\k)_{m,n} = z_{m,n}(e^{in\k Q_0}\otimes \1),
\]
where $z = (z_{m,n})$ is some element in $\M$. Similarly,
an element $z^\prime_\k\in \M_{Q,\k}'$ has the components of the form
\[
(z^\prime_\k)_{m,n} = z_{m,n}(\1 \otimes e^{im\k Q_0}),
\]
where $z' = (z_{m,n})$ is some element in $\M'$.
\end{lemma}
\begin{proof}
We will show only the former statement since the latter is analogous.
First we consider an element of a simple form $(x_m\otimes \1)S(\1\otimes y_n)S^*$,
where $x_m\in \A_0(\RR_-)_m$ and $y_n \in \A_0(\RR_+)_n$. We saw in
Proposition \ref{lm:adjoint-action} that this is equal to
$(x_m\otimes y_n)(e^{i\k nQ_0}\otimes \1)$, thus this has the asserted form.
Note that the linear space spanned by these elements for different $m,n$
is closed even under
product. For a finite product and sum, the thesis is linear with respect to
$x$ and $y$, hence we obtain the desired decomposition.
The von Neumann algebra $\M_{Q_0,\k}$ is linearly generated by these elements.
Recalling that $z_{m,n}$ is a matrix element with respect to the decomposition
$\1 = \sum_{l,l'} E(l,l')$, we obtain the Lemma.
\end{proof}

Now we are going to determine the intersection of wedge algebras.
At this point, we need to use unexpectedly strong additivity and
conformal covariance (see Section \ref{preliminaries}).
The fixed point subnet $\A_0^{\a_0}$ of a strongly additive net $\A_0$
on $S^1$ with respect to the action $\a_0$ of a compact group $S^1$
of inner symmetry is again strongly additive \cite{Xu05}.

\begin{example}
The loop group nets $\A_{SU(N),k}$ are completely rational \cite{GF93,Xu00-1}, hence
in particular they are strongly additive. Moreover, they are conformal \cite{PS86}.

If $\A_0$ is diffeomorphism covariant,
the strong additivity follows from the split property and the finiteness of $\mu$-index \cite{LX04}.
We have plenty of examples of nets which satisfy strong additivity and conformal covariance
since it is known that complete rationality passes to
finite index extensions and finite index subnets \cite{Longo03}.
\end{example}

\begin{theorem}\label{th:intersection}
Let $\A_0$ be strongly additive and conformal and
$e^{is Q_0}$ implement a periodic family of inner symmetries with the
generator $Q_0$.
We write, with a little abuse of notation, $T(t_+,t_-) := T_0(t_+)\otimes T_0(t_-)$.
For $t_+ < 0$ and $t_- > 0$ we have
\[
\M_{Q_0,\k} \cap \left(\Ad T(t_+,t_-) (\M_{Q_0,\k}^\prime)\right)
= \A_0^G((t_+,0))\otimes\A_0^G((0,t_-)),
\]
where $G$ is the group of automorphisms generated by $\Ad e^{i\k Q_0}$.
\end{theorem}
\begin{proof}
Let us consider an element from the intersection. From Lemma \ref{lm:component},
we have two descriptions of such an element, namely,
\begin{eqnarray*}
(z_\k)_{m,n}
&=& z_{m,n}(e^{in\k Q_0}\otimes \1), \,\,\, z \in \A_0(\RR_-)\otimes \A_0(\RR_+),\\
(z^\prime_\k)_{m,n}
&=& z^\prime_{m,n}(\1\otimes e^{im\k Q_0}), \,\,\, z^\prime \in \A_0(\RR_+ + t_+)\otimes \A_0(\RR_- + t_-).
\end{eqnarray*}
If these elements have to coincide, each $(m,n)$ component has to coincide. Or
equivalently, it should happen that
$z_{m,n}(e^{in\k Q_0}\otimes e^{-im\k Q_0})= z^\prime_{m,n}$.

Recall that an inner symmetry commutes with diffeomorphisms \cite{CW05}.
This implies that the fixed point subalgebra contains the representatives of
diffeomorphisms. Furthermore, the fixed point subalgebra by a compact group
is again strongly additive \cite{Xu05}. This means that
\begin{align*}
\A_0^{\a_0}((-\infty,t_+))\vee\A_0^{\a_0}((0,\infty)) = \A_0^{\a_0}((t_+,0)'),\\
\A_0^{\a_0}((-\infty,0))\vee\A_0^{\a_0}((t_-,\infty)) = \A_0^{\a_0}((0,t_-)'),
\end{align*}
where $I'$ means the complementary interval in $S^1$.

We claim that if for $z\in \A_0(\RR_-)\otimes \A_0(\RR_+)$ and
$z' \in \A_0(\RR_+ + t_+)\otimes \A_0(\RR_- + t_-)$ there holds
$z\cdot (e^{im\k Q}\otimes e^{-in\k Q}) = z'$,
then $z = z' \in \left(\A_0(\RR_-)\cap\A_0(\RR_+ + t_0)\right) \otimes \left(\A_0(\RR_+)\cap\A_0(\RR_- + t_-)\right)$.
Indeed, since $z \in \A_0(\RR_-)\otimes \A_0(\RR_+)$, it commutes with $U(g_+\times g_-)$
with $\supp(g_+)\subset \RR_+$ and $\supp(g_-)\subset \RR_-$. Similarly, 
$z' \in \A_0(\RR_+ + t_0)\otimes \A_0(\RR_- + t_1)$ commutes with $U(g_+\times g_-)$
with $\supp(g_+)\subset \RR_- + t_+$ and $\supp(g_-)\subset \RR_+ + t_-$.
Furthermore, the unitary $e^{im\k Q}\otimes e^{-in\k Q}$ which implements
an inner symmetry commutes with any action of diffeomorphism \cite{CW05}.
Recall that the fixed point subalgebra is strongly additive, hence
by the assumed equality $z\cdot (e^{im\k Q}\otimes e^{-in\k Q}) = z'$, this element
commutes with $\A_0^{\a_0}((t_+,0)')\otimes\A_0^{\a_0}((0,t_-)')$.
In particular, it commutes with any diffeomorphism of $S^1\times S^1$
supported in $(t_+,0)'\times(0,t_-)'$.
There is a sequence of diffeomorphisms $g_i$ which take
$\RR_-\times \RR_+$ to $(t_+-\e_i,0)\times(0,t_-+\e_i)$ with support disjoint from
$(t_+,0)\times (0,t_-)$ for arbitrary small $\e_i > 0$. This fact and
the diffeomorphism covariance imply that $z$ is indeed
contained in $\A_0((t_+,0))\otimes\A_0((0,t_-))$.
By a similar reasoning, one sees that $z' \in \A_0((t_+,0))\otimes\A_0((0,t_-))$ as well.
Now by Reeh-Schlieder property for $\A_0((t_+,0))\otimes\A_0((0,t_-))$ we have $z = z'$ since
$z \Omega = z\cdot (e^{im\k Q}\otimes e^{-in\k Q})\Omega = z'\Omega$.


Thus, if $z_{m,n}(e^{im\k Q}\otimes e^{-in\k Q})= z^\prime_{m,n}$, then
$z_{m,n} = z^\prime_{m,n} \in \A_0((t_+,0))\otimes\A_0((0,t_-))$. Furthermore,
by Corollary \ref{cr:separating}, there is a separating vector $v \in E(l,l')\H$.
Now it holds that $e^{inl\k-iml'\k}z_{m,n}v = z'_{m,n}v$, hence from
the separating property of $v$ it follows that $e^{inl\k-iml'\k}z_{m,n} = z'_{m,n}$
for each pair $(l,l') \in \ZZ\times\ZZ$. This is possible only if both $n\k$ and $m\k$
are $2\pi$ multiple of an integer or $z_{m,n} = z'_{m,n} = 0$. This is equivalent to
that $\Ad e^{i\k mQ_0}\otimes e^{i\k nQ_0}(z) = z$, namely, $z$ is an element of
the fixed point algebra $\A_0^G((t_+,0))\otimes \A_0^G((0,t_-))$
by the action $\Ad e^{i\k mQ_0}\otimes e^{i\k nQ_0}$ of $G\times G$.
\end{proof}

Note that the size of the intersection is very sensitive to the parameter $\k$:
If $\k$ is $2\pi$-multiple of a rational number, then the inclusion
$[\A_0,\A_0^G]$ has finite index. Otherwise, it has infinite index.

Finally, we comment on the net generated by the intersection.
The intersection takes a form of chiral net $\A_0^G\otimes \A_0^G$
where $G$ is generated by $\Ad e^{i\k Q_0}$,
hence the S-matrix is trivial \cite{DT11-1}. This result is expected also from \cite{Tanimoto11-2},
where M\"obius covariant net has always trivial S-matrix. Our construction is based on
inner symmetries which commute with M\"obius symmetry, hence the net of strictly local
elements is necessarily M\"obius covariant, then it should have trivial S-matrix.
But from this simple argument one cannot infer that the intersection should be
asymptotically complete, or equivalently chiral. This exact form of the intersection
can be found only by the present argument.

\subsubsection{Construction through cyclic group actions}\label{cyclic}
Here we briefly comment on the actions by the cyclic group $\ZZ_k$.
In previous Sections, we have constructed Borchers triples for the action of $S^1$.
It is not difficult to replace $S^1$ by a finite group $\ZZ_k$. Indeed,
the main ingredient was the existence of the Fourier components.
For $\ZZ_k$-actions, the discrete Fourier transform is available
and all the arguments work parallelly (or even more simply). For the later use,
we state only the result without repeating the obvious modification of definitions and proofs.
\begin{theorem}\label{th:triple-cyclic}
Let $\A_0$ be a strongly additive conformal net on $S^1$ and $\a_{0,n} = \Ad e^{i\frac{2\pi n}{k}Q_0}$ be an action
of $\ZZ_k$ as inner symmetries. Then, for $n \in \ZZ_k$, the triple
\begin{eqnarray*}
\M_{Q_0,n} &:=& \{x\otimes \1, \Ad e^{i\frac{2\pi n}{k} Q_0\otimes Q_0}(\1\otimes y):
 x \in \A_0(\RR_-), y \in \A(\RR_+)\}'' \\
T &:=& T_0\otimes T_0 \\
\Omega &:=& \Omega_0\otimes \Omega_0
\end{eqnarray*}
is an asymptotically complete Borchers triple with S-matrix $e^{i\frac{2\pi n}{k}Q_0\otimes Q_0}$.
As for strictly local elements, we have
\[
\M_{Q_0,n} \cap \left(\Ad T(t_+,t_-) (\M_{Q,n}^\prime)\right)
= \A_0^G((t_+,0))\otimes\A_0^G((0,t_-)),
\]
where $G$ is the group of automorphisms of $\A_0$ generated by $\Ad e^{i\frac{2\pi n}{k}Q_0}$.
\end{theorem}
Note that, although the generator $Q_0$ of inner symmetries
of the cyclic group $Z_k$ is not unique, we used it always in
the form $e^{i\k Q_0}$ or $e^{i\k Q_0\otimes Q_0}$ and these
operators are determined by the automorphisms.
Spectral measures can be defined in terms of these
exponentiated operators uniquely on (the dual of) the
cyclic group $\ZZ_k$. In this way, the choice of $Q_0$
does not appear in the results and proofs.

\section{Construction through endomorphisms on the $U(1)$-current net}\label{u1-current}
\subsection{The $U(1)$-current net and Longo-Witten endomorphisms}\label{u1-current-summary}
In this Section we will construct a family of Borchers triples for a
specific net on $S^1$. Since we need explicit formulae for the relevant operators,
we briefly summarize here some facts about the net called
the $U(1)$-current net, or the (chiral part of) free massless bosonic field.
On this model, there has been found a family of Longo-Witten endomorphisms \cite{LW11}.
We will construct a Borchers triple for each of these endomorphisms.
This model has been studied with the algebraic approach since the fundamental paper \cite{BMT88}.
We refer to \cite{Longo08} for the notations and the facts in the following.

A fundamental ingredient is the irreducible unitary representation of the M\"obius group
with the lowest weight $1$: Namely, we take the irreducible
representation of $\psl2r$ of which the
smallest eigenvalue of the rotation subgroup is $1$.
We call the Hilbert space $\H^1$. We take a specific realization of this representation.
Namely, let $C^\infty(S^1,\RR)$ be the space of real-valued smooth functions on $S^1$.
This space admits a seminorm
\[
\|f\| := \sum_{k\ge 0} 2k|\hat{f}_k|^2,
\]
where $\hat{f}_k$ is the $k$-th Fourier component of $f$,
and a complex structure
\[
(\widehat{\I f})_k = -i \mathrm{sign}(k) \hat{f}_k.
\] 
Then, by taking the quotient space by the null space with respect to the seminorm,
we obtain the complex Hilbert space $\H^1$. We say $C^\infty(S^1,\RR) \subset \H^1$.
On this space, there acts $\psl2r$ by naturally extending the action on $C^\infty(S^1,\RR)$.

Let us denote $\H^n := \H^{\otimes n}$ for a nonnegative integer $n$. On this space,
there acts the symmetric group $\mathrm{Sym}(n)$. Let $Q_n$ be the projection onto the invariant subspace
with respect to this action. We put $\H^n_s := Q_n\H^n$ and the {\bf symmetric Fock space}
\[
\H^\sc_s:= \bigoplus_n \H^n_s,
\]
and this will be the Hilbert space of the $U(1)$-current net on $S^1$.
For $\xi \in \H^1$, we denote by $e^\xi$ a vector of the form
$\sum_n \frac{1}{n!}\xi^{\otimes n} = 1 \oplus \xi \oplus \left(\frac{1}{2}\xi\otimes\xi \right)\oplus\cdots$.
Such vectors form a total set in $\H^\sc_s$.
The {\bf Weyl operator} of $\xi$ is defined by
$W(\xi)e^{\eta} = e^{-\frac{1}{2}\<\xi,\xi\>-\<\xi,\eta\>} e^{\xi+\eta}$.

The Hilbert space $\H^\sc_s$ is naturally included in the {\bf unsymmetrized Fock space}:
\[
\H^\sc := \bigoplus_n (\H^1)^{\otimes n} = \CC \oplus \H^1 \oplus \left(\H^1\otimes \H^1\right) \oplus \cdots
\]
We denote by $Q_\sc$ the projection onto $\H^\sc_s$.
For an operator $X_1$ on the one particle space $\H^1$,
we define the {\bf second quantization} of $X_1$ on $\H^\sc_s$ by
\[
\G(X_1) := \bigoplus_n (X_1)^{\otimes n} = 1 \oplus X_1 \oplus \left(X_1\otimes X_1\right)\oplus \cdots
\]
Obviously, $\G(X_1)$ restricts to the symmetric Fock space $\H^\sc_s$. We still write this restriction
by $\G(X_1)$ if no confusion arises.
For a unitary operator $V_1 \in B(\H^1)$ and $\xi \in \H^1$, it holds that
$\G(V_1)e^\xi = e^{V_1\xi}$ and $\Ad \G(V_1)(W(\xi)) = W(V_1\xi)$.
On the one particle space $\H^1$, there acts the M\"obius group $\psl2r$ irreducibly by $U_1$.
Then $\psl2r$ acts on $\H^\sc$ and on $\H^\sc_s$ and by $\G(U_1(g))$, $g\in \psl2r$.
The representation of the translation subgroup in $\H^1$ is denoted by $T_1(t) = e^{itP_1}$
with the generator $P_1$.

The {\bf $U(1)$-current net} $\u1$ is defined as follows:
\[
\u1(I) := \{W(f): f\in C^\infty(S^1,\RR) \subset \H^1, \supp(f) \subset I\}''.
\]
The vector $1 \in \CC = \H^0 \subset \H^\sc_s$ serves as the vacuum vector $\Omega_0$
and $\G(U_1(\cdot))$ implements the M\"obius symmetry.
We denote by $T^\sc_s$ the representation of one-dimensional translation of
$\u1$.

For this model, a large family of endomorphisms has been found by Longo and Witten.
\begin{theorem}[\cite{LW11}, Theorem 3.6]\label{th:longo-witten}
Let $\f$ be an inner symmetric function on the upper-half plane $\SS_\infty \subset \CC$:
Namely, $\f$ is a bounded analytic function of $\SS_\infty$ with the boundary value
$|\f(p)| = 1$ and $\f(-p) = \overline{\f(p)}$ for $p \in \RR$.
Then $\G(\f(P_1))$ commutes with $T^\sc_s$ (in particular $\G(\f(P_1))\Omega_0 = \Omega_0$)
and $\Ad \G(\f(P_1))$ preserves $\u1(\RR_+)$. In other words, $\G(\f(P_1))$ implements
a Longo-Witten endomorphism of $\u1$.
\end{theorem}

\subsection{Construction of Borchers triples}\label{u1-net}
In this Section, we construct a Borchers triple for a fixed $\f$,
the boundary value of an inner symmetric function (see Section \ref{u1-current-summary}).
Many operators are naturally defined on the unsymmetrized Fock space, hence
we always keep in mind the inclusion $\H^\sc_s \subset \H^\sc$.
The full Hilbert space for the two-dimensional Borchers triples will be
$\H^\sc_s \otimes \H^\sc_s$.

On $\H^m$, there act $m$ commuting operators
\[
\{\1\otimes\cdots\otimes \underset{i\mbox{-th}}{P_1}\otimes\cdots \otimes \1: 1\le i \le m\}.
\]
We construct a unitary operator by the functional calculus on the corresponding spectral measure.
We set
\begin{itemize}
\item $P_{i,j}^{m,n} := (\1\otimes\cdots\otimes \underset{i\mbox{-th}}{P_1}\otimes\cdots \otimes \1)
\otimes (\1\otimes\cdots\otimes \underset{j\mbox{-th}}{P_1}\otimes\cdots \otimes \1)$,
which acts on $\H^m\otimes\H^n$, $1 \le i \le m$ and $1 \le j \le n$.
\item $S^{m,n}_\f := \prod_{i,j}\f(P_{i,j}^{m,n})$, where
$\f(P_{i,j}^{m,n})$ is the functional calculus on $\H^m\otimes\H^n$.
\item $S_\f := \bigoplus_{m,n} S_\f^{m,n} = \bigoplus_{m,n} \prod_{i,j} \f(P^{m,n}_{i,j})$
\end{itemize}
By construction, the operator $S_\f$ acts on $\H^\sc\otimes\H^\sc$. Furthermore, it is easy to see
that $S_\f$ commutes with both $Q_\sc\otimes \1$ and $\1\otimes Q_\sc$: In other words, $S_\f$ naturally
restricts to partially symmetrized subspaces $\H^\sc_s\otimes\H^\sc$ and $\H^\sc\otimes\H^\sc_s$
and to the totally symmetrized space $\H^\sc_s\otimes\H^\sc_s$.
Note that $S^{m,n}_\f$ is a unitary operator on the Hilbert spaces $\H^m\otimes\H^n$
and $S_\f$ is the direct sum of them.

Let $E_1\otimes E_1\otimes\cdots\otimes E_1$ be the joint spectral measure
of operators $\{\1\otimes\cdots\otimes \underset{j\mbox{-th}}{P_1}\otimes\cdots \otimes \1: 1\le j \le n\}$.
The operators $\{\f(P_{i,j}^{m,n}):1\le i \le m, 1\le j \le n\}$ and $S_\f^{m,n}$ are compatible with
the spectral measure
$\left(\overset{m\mbox{-}\mathrm{times}}{\overbrace{E_1\otimes E_1\otimes\cdots\otimes E_1}}\right)
\otimes\left(\overset{n\mbox{-}\mathrm{times}}{\overbrace{E_1\otimes E_1\otimes\cdots\otimes E_1}}\right)$
and one has\[
\f(P_{i,j}^{m,n}) = \int \left(\1\otimes\cdots \otimes\underset{i\mbox{-th}}{\f(p_jP_1)}\otimes \cdots \1\right)
\otimes \left(\1\otimes\cdots \underset{j\mbox{-th}}{dE_1(p_j)}\otimes\cdots\1\right).
\]
For $m=0$ or $n=0$ we set $\f_{i,j}^{m,n} = \1$ as a convention.

According to this spectral decomposition, we decompose $S_\f$ with respect only to
the right component as in the commutativity Lemma \ref{lm:commutativity}:
\begin{eqnarray*}
S_\f &=& \bigoplus_{m,n}\prod_{i,j}\f(P_{i,j}^{m,n})\\
&=& \bigoplus_{m,n} \prod_{i,j} \int \left(\1\otimes\cdots\otimes\underset{i\mbox{-th}}{\f(p_jP_1)}\otimes\cdots\1\right)
 \otimes dE_0(p_1)\otimes\cdots\otimes dE(p_n) \\
&=& \bigoplus_{m,n} \int \prod_{i,j}  \left(\1\otimes\cdots\otimes\underset{i\mbox{-th}}{\f(p_jP_1)}\otimes\cdots\1\right)
 \otimes dE_0(p_1)\otimes\cdots\otimes dE(p_n) \\
&=& \bigoplus_n\int \bigoplus_m \prod_j  (\f(p_jP_1))^{\otimes m} \otimes dE_1(p_1)\otimes\cdots\otimes dE_1(p_n) \\
&=& \bigoplus_n\int \prod_j \bigoplus_m (\f(p_jP_1))^{\otimes m} \otimes dE_1(p_1)\otimes\cdots\otimes dE_1(p_n) \\
&=& \bigoplus_n\int \prod_j \G(\f(p_jP_1)) \otimes dE_1(p_1)\otimes\cdots\otimes dE_1(p_n),
\end{eqnarray*}
where the integral and the product commute in the third equality since the spectral measure
is disjoint for different values of $p$'s, and the sum and the product commute in the fifth equality
since the operators in the integrand act on mutually disjoint spaces, namely on $\H^m\otimes\H^\sc$
for different $m$. Since all operators appearing in the integrand in the last expression
are the second quantization
operators, this formula naturally restricts to the partially symmetrized space $\H^\sc_s\otimes\H^\sc$.

\begin{lemma}\label{lm:u1-commutativity}
It holds for $x\in\u1(\RR_-)$ and $x' \in \u1(\RR_+)$ that
\[
[x\otimes \1, \Ad S_\f (x'\otimes\1)] = 0,
\]
on the Hilbert space $\H^\sc_s\otimes\H^\sc_s$.
\end{lemma}
\begin{proof}
The operator $S_\f$ is disintegrated into second quantization operators as we saw above.
If $\f$ is an inner symmetric function, then so is $\f(p_j\cdot)$, $p_j\ge0$, thus
each $\G(\f(p_jP_1))$ implements a Longo-Witten endomorphism.

Note that $S_\f$ restricts naturally to $\H^\sc_s\otimes\H^\sc$ by construction and $x\otimes\1$
and $x'\otimes\1$ extend naturally to $\H^\sc_s\otimes\H^\sc$ since the right-components
of them are just the identity operator $\1$.
Then we calculate the commutation relation on $\H^\sc_s\otimes\H^\sc$. 
This is done in the same way as Lemma \ref{lm:commutativity}: Namely, we have
\[
\Ad S_\f(x'\otimes\1)
= \bigoplus_n\int \Ad \left(\prod_j \G(\f(p_jP_1))\right)(x') \otimes dE_1(p_1)\otimes\cdots\otimes dE_1(p_n).
\]
And this commutes with $x\otimes\1$. Indeed, since $x\in\u1(\RR_-)$ and $x'\in\u1(\RR_+)$, hence
$\Ad \G(\f(p_j))(x') \in \u1(\RR_+)$ for any $p_j \ge 0$ by Theorem \ref{th:longo-witten} of
Longo-Witten, and by the fact that the spectral support of $E_1$ is positive.
Precisely, we have $[x\otimes \1, \Ad S_\f (x'\otimes\1)] = 0$ on $\H^\sc_s\otimes\H^\sc$.

Now all operators $S_\f$, $x\otimes \1$ and $x'\otimes \1$ commute with $\1\otimes Q_\sc$,
we obtain the thesis just by restriction.
\end{proof}

Finally we construct a Borchers triple by following the prescription at the end
of Section \ref{complete-invariant}.
\begin{theorem}
The triple
\begin{itemize}
\item $\M_\f := \{x\otimes \1, \Ad S_\f (\1\otimes y): x\in\u1(\RR_-),y\in\u1(\RR_+)\}''$
\item $T = T_0\otimes T_0$ of $\u1\otimes\u1$
\item $\Omega = \Omega_0\otimes \Omega_0$ of $\u1\otimes\u1$
\end{itemize}
is an asymptotically complete Borchers triple with S-matrix $S_\f$.
\end{theorem}
\begin{proof}
This is almost a repetition of the proof of Theorem \ref{th:triple-translation}.
Namely, the conditions on $T$ and $\Omega$ are readily satisfied since 
they are same as the chiral triple. The operators $S_\f$ and $T$ commute since
both are the functional calculus of the same spectral measure, hence
$T(t_0,t_1)$ sends $\M_\f$ into itself for $(t_0,t_1) \in W_\R$.
The vector $\Omega$ is cyclic for $\M_\f$ since
$\M_\f\Omega \supset \{x\otimes \1\cdot S_\f\cdot \1\otimes y\cdot\Omega\} = \{x\otimes \1\cdot \1\otimes y\cdot\Omega\}$
and the latter is dense by the Reeh-Schlieder property of the chiral net.
The separating property of $\Omega$ is shown through Lemma \ref{lm:u1-commutativity}.
\end{proof}
\begin{remark}
In this approach, the function $\f$ itself appears in two-particle scattering,
not the square as in \cite{Lechner11}. Thus, although the formulae look similar,
the present construction contains much more examples.
\end{remark}
\subsubsection*{Intersection property for constant functions $\f$}
For the simplest cases $\f(p) = 1$ or $\f(p) = -1$, we can easily determine the
strictly local elements. Indeed, for $\f(p) = 1$, $S_\f = \1$ and
the Borchers triple coincides with the one from the original chiral net.
For $\f(p) = -1$, $S_\f^{m,n} = (-1)^{mn}\cdot \1$ and it is not difficult to see that
if one defines an operator $Q_0 := 2P_e - \1$, where $P_e$ is a projection onto
the ``even'' subspace $\bigoplus_n \H^{2n}_s$ of $\H^\sc_s$, then
$e^{i\pi Q_0}$ implements a $\ZZ_2$-action of inner symmetries on $\u1$ and
$S_\f = e^{i\pi Q_0\otimes Q_0}$. Then Theorem \ref{th:triple-cyclic} applies
to find that the strictly local elements are of the form
$\u1^{\ZZ_2}\otimes\u1^{\ZZ_2}$ where the action of $\ZZ_2$ is realized by
$\Ad e^{i\pi n Q_0}$.
\subsection{Free fermionic case}
As explained in \cite{LW11}, one can construct a family of endomorphisms
on the Virasoro net $\vir_c$ with the central charge $c = \frac{1}{2}$
by considering the free fermionic field. With a similar construction using
the one-particle space on which the M\"obius group acts irreducibly and
projectively with the lowest weight $\frac{1}{2}$, one considers the free fermionic
(nonlocal) net on $S^1$, which contains $\vir_{\frac{1}{2}}$ with index $2$.

The endomorphisms are implemented again by the second quantization operators.
By ``knitting up'' such operators as is done for bosonic $U(1)$-current case,
then by restricting to the observable part $\vir_\frac{1}{2}$, 
we obtain a family of Borchers triples with the asymptotic algebra
$\vir_\frac{1}{2}\otimes \vir_\frac{1}{2}$ with nontrivial S-matrix.
In the present article we omit the detail, and hope to return to this subject with
further investigations.

\section{Conclusion and outlook}\label{conclusion}
We showed that any two-dimensional massless asymptotically complete model
is characterized by its asymptotic algebra which is automatically a chiral M\"obius net,
and the S-matrix. Then we reinterpreted the Buchholz-Lechner-Summers deformation applied
to chiral conformal net in this framework:
It corresponds to the S-matrix $e^{i\k P_0\otimes P_0}$.
Furthermore we obtained wedge-local nets through periodic inner symmetries which have
S-matrix $e^{i\k Q_0\otimes Q_0}$. We completely determined the strictly local contents
in terms of the fixed point algebra when the chiral component is strongly additive and conformal.
Unfortunately, the S-matrix restricted to the strictly local part is trivial.
For the $U(1)$-current net and the Virasoro net $\vir_c$ with $c=\frac{1}{2}$,
we obtained families of wedge-local nets parametrized by inner symmetric functions $\f$.

One important lesson is that construction of wedge-local nets should be considered as an intermediate step
to construct strictly local nets: Indeed, any M\"obius covariant net has trivial S-matrix \cite{Tanimoto11-2},
hence the triviality of S-matrix in the construction through inner symmetries
is interpreted as a natural consequence. Although the S-matrix as a Borchers triple is nontrivial,
this should be treated as a false-positive. The true nontriviality should be inferred by
examining the strictly local part. On the other hand, we believe that
the techniques developed in this paper will be of importance in the further explorations
in strictly local nets. The sensitivity of the strictly local part to the parameter
$\k$ in the case of the construction with respect to inner symmetries gives another
insight.

Apart from the problem of strict locality, a more systematic study of the necessary
or sufficient conditions for S-matrix is desired. Such a consideration could lead
to a classification result of certain classes of massless asymptotically complete models.
For the moment, a more realistic problem would be to construct S-matrix with the asymptotic algebra
$\A_N\otimes\A_N$, where $\A_N$ is a local extension of the $U(1)$-current net \cite{BMT88, LW11}.
A family of Longo-Witten endomorphisms has been constructed also for $\A_N$, hence
a corresponding family of wedge-local net is expected and recently a similar kind of endomorphisms
has been found for a more general family of nets on $S^1$ \cite{Bischoff11}.
Or a general scheme of deforming a given Wightman-field theoretic net has been established \cite{Lechner11}.
The family of S-matrices constructed in the present paper seems rather small, since
there is always a pair of spectral measures and their tensor product diagonalizes the S-matrix.
This could mean in physical terms that the interaction between two waves is not very strong.
We hope to address these issues in future publications.

\subsubsection*{Acknowledgment.}
I am grateful to Roberto Longo for his constant support.
I wish to thank Marcel Bischoff, Kenny De Commer, Wojciech Dybalski and Daniele Guido
for useful discussions.

\appendix
\newcommand{\appsection}[1]{\let\oldthesection\thesection
  \renewcommand{\thesection}{Appendix \oldthesection}
  \section{#1}\let\thesection\oldthesection}
\appsection{A remark on BLS deformation}\label{a-remark}
In the proof of Theorem \ref{th:coincidence} we used the fact that
$(x\otimes \1)_{\Theta_\k} \xi\otimes \Omega_0 = x\xi\otimes \Omega_0$.
The equation (2.2) from \cite{BLS} translates in our notation to
\[
(x\otimes\1)_{\Theta_\k} = \underset{F\nearrow \1}{\lim_{B\nearrow \RR^2}}\int_B \Ad U(\k t_1, \k t_0)(x\otimes \1)FdE(t_0,t_1),
\]
where $B$ is a bounded subset in $\RR^2$ and $F$ is a finite dimensional subspace
in $\H$. Now it is easy to see that
$(x\otimes \1)_{\Theta_\k}(\xi\otimes \Omega_0) = x\xi\otimes \Omega_0$.
Indeed, we have $\xi\otimes\Omega_0 \in E(\L_+)$, where $\L_+ = \{(p_0,p_1)\in \RR^2:p_0+p_1=0\}$,
hence the integral above is concentrated in $\L_+$, and for $(t,t)\in\L_+$
it holds that $\Ad U(\k t,\k t)(x\otimes \1) = \Ad \1\otimes U_0(\sqrt{2}\k t)(x\otimes\1) = x\otimes \1$.
Then the integral simplifies as follows:
\begin{eqnarray*}
(x\otimes \1)_{\Theta_\k}(\xi\otimes \Omega_0)
&=& \underset{F\nearrow \1}{\lim_{B\nearrow \RR^2}}\int_{B\cap \L_+}
 \Ad U(\k t, \k t)(x\otimes \1)
 \cdot F\cdot dE(t,t)(\xi\otimes \Omega_0)\\
&=& \underset{F\nearrow \1}{\lim_{B\nearrow \RR^2}}\int_{B\cap \L_+}
 x\otimes\1 \cdot dE(t,t)(\xi\otimes\Omega_0)\\
&=& x\xi\otimes\Omega_0.
\end{eqnarray*}
This is what we had to prove.

\def\cprime{$'$}

\end{document}